\def\bargr{\mathop{{\overline{\rm gr}}}}
\def\rltwo{\mathop{\rho_{l^2}}}
\def\id{\mathop{\mathrm{id}}} 
\def\od{\mathop{\mathrm{od}}} 
\def\deg{\mathop{\mathrm{deg}}} 
\def\Ep{\mathbb{E}_{\{p\}}}
\def\Pp{\mathbb{P}_{\{p\}}}

\def\gr{\mathop{\mathrm{gr}}}
\def\br{\mathop{\mathrm{br}}}

\documentclass[preprint,5p,times,letterpaper,sort&compress]{elsarticle}
\usepackage[pdftex]{hyperref}
\usepackage{amsfonts}
\usepackage{amssymb}
\usepackage{graphicx}
\usepackage{color}
\usepackage{amsthm}
\usepackage{amsmath}
\newtheorem{example}{Example}
\newtheorem{theorem}{Theorem}
\newtheorem{lemma}[theorem]{Lemma}
\newtheorem{corollary}[theorem]{Corollary}
\newtheorem{statement}[theorem]{Statement}

\begin{document}
\begin{frontmatter}

\title{Algebraic bounds for heterogeneous site percolation on
  directed and undirected graphs}

\author[ORNLaddress]{Kathleen E. Hamilton\corref{correspondingauthor}\fnref{ORNLdisclaimer}}
\ead{hamiltonke@ornl.gov}
\author[UCRaddress]{Leonid P. Pryadko}
\ead{leonid.pryadko@ucr.edu}

\cortext[correspondingauthor]{Corresponding author}
\fntext[ORNLdisclaimer]{This submission was written by the author
  acting in her own independent capacity and not on behalf of
  UT-Battelle, LLC, or its affiliates or successors.}

\address[ORNLaddress]{Quantum Computing Institute, Oak Ridge National
  Laboratory, Oak Ridge, Tennessee, 37821, USA}
\address[UCRaddress]{Department of Physics \& Astronomy, University of
  California, Riverside, California, 92521, USA}

\begin{abstract}
  We analyze site percolation on directed and undirected graphs with
  site-dependent open-site probabilities. We construct upper bounds on
  cluster susceptibilities, vertex connectivity functions, and the
  expected number of simple open cycles through a chosen arc; separate
  bounds are given on finite and infinite (di)graphs.  These produce
  lower bounds for percolation and uniqueness transitions in infinite
  (di)graphs, and for the formation of a giant component in finite
  (di)graphs.  The bounds are formulated in terms of appropriately
  weighted adjacency and non-backtracking (Hashimoto) matrices.  It
  turns out to be the uniqueness criterion that is most closely
  associated with an asymptotically vanishing probability of forming a
  giant strongly-connected component on a large finite (di)graph.
\end{abstract}

\end{frontmatter}

\section{Introduction}

We are currently living in an age where many scientific
and industrial applications rapidly generate large datasets. The
connectivity and underlying structure of this data is of great
interest. As a result, graph theory has enjoyed a resurgence, becoming
a prominent tool for describing complex connections in various kinds
of networks: social, biological,
technological\cite{Albert-Jeong-Barabasi-2000,
  Albert-Barabasi-RMP-2002,%
  Borner-Sanyal-Vspignani-ARIST-2007,%
  Danon-etal-2011,Costa-Oliveira-CorreaRocha-2011,%
  PastorSatorras-etal-RMP-2015,Radicchi-2015,Radicchi-interdep-2015},
and many others. Percolation on graphs has been used to describe
internet stability\cite{Cohen-Erez-benAvraham-Havlin-2000,%
  Callaway-Newman-Strogatz-Watts-2000}, spread of contagious
diseases\cite{Grassberger-1983,Moore-Newman-2000,Sander-epidemics-2002}
and computer viruses\cite{PastorSatorras-Vespignani-2001}; related
models describe market crashes\cite{Gai-Kapadia-2010} and viral spread
in social networks \cite{Watts-PNAS-2002,Kempe-Kleinberg-Tardos-2003,%
  Jiang-Miao-Yi-Zhenzhong-Hauptmann-2014}. General percolation theory
methods are increasingly used in quantum information
theory\cite{Delfosse-Zemor-2012,Kovalev-Pryadko-FT-2013,%
  PhysRevA.69.062311,PhysRevLett.99.130501,PhysRevLett.103.240503}.
Percolation is also an important phase transition in its own
right\cite{Kasteleyn-Fortuin-1969,Fortuin-Kasteleyn-1972,%
  Essam-1980,Wu-RMP-1982} and is well established in physics as an
approach for dealing with strong disorder: quantum or classical
transport\cite{Ambegaokar-Halperin-Langer-1971,Kirkpatrick-1973,%
  isichenko-rmp-1992}, bulk properties of composite
materials\cite{Bergman-Imry-1977,Nan-Shen-Ma-2010}, diluted magnetic
tran\-sitions\cite{Stinchcombe-1983}, or spin glass transitions%
\cite{Almeida-1978,Novak-1985,Arcangelis-1991,Coniglio-1994,%
  PhysRevLett.93.040401}.

Recently, we suggested\cite{Hamilton-Pryadko-PRL-2014} a lower bound
on the site percolation transition on an infinite 
graph,
\begin{equation}
  p_c\ge 1/\rho(H).
  \label{eq:threshold-H}
\end{equation}
Here,
$\rho(H)$ is the spectral radius of the non-backtracking (Hashimoto)
matrix\cite{Hashimoto-matrix-1989} $H$ associated with the graph.
This expression has been proved\cite{Hamilton-Pryadko-PRL-2014} for
infinite \emph{quasi-transitive} graphs, a graph-theoretic analog of
translationally-invariant systems with a finite number of inequivalent
sites.  The bound (\ref{eq:threshold-H}) is achieved on any infinite
quasi-transitive tree\cite{Hamilton-Pryadko-PRL-2014}, and it also
gives numerically exact percolation thresholds for several families of
large random graphs, as well as for some large empirical
graphs\cite{Karrer-Newman-Zdeborova-PRL-2014}.%

In most applications of percolation theory, one encounters large, but
finite, graphs.  The expectation is that the corresponding
crossover retains some properties of the transition in the infinite
graphs, e.g., the formation of large open clusters be unlikely if the
open site probability $p$ is well below $p_c$.  However, the bound
(\ref{eq:threshold-H}) tells nothing about the structure of the
percolating cluster on finite graphs, and neither it gives an
algorithm for computing the location of the crossover in the case of a
finite graph\cite{Radicchi-2015}.  In particular,
Eq.~(\ref{eq:threshold-H}) misses the mark entirely for any finite
tree where $\rho(H)=0$.

In this work, we construct several spectral and algebraic bounds for
transitions associated with heterogeneous site percolation on directed
and undirected graphs, both finite and infinite, and analyze the
continuity of these bounds for a sequence of finite digraphs weakly
convergent to an infinite graph.  Namely, for finite digraphs, we
construct explicit upper bounds for the local in-/out-/strong-cluster
susceptibilities (average size of a cluster connected to a given
site), the strong connectivity function (probability that a given pair
of sites belongs to the same strongly-connected cluster), and the
expected number of simple cycles passing through a given arc.  We also
construct some analogous bounds for infinite digraphs, which result in
non-trivial lower bounds for the transitions associated with divergent
in-/out-cluster susceptibilities, emergence of infinite
in-/out-clusters, and the strong-cluster uniqueness transition.

Our results imply that Eq.~(\ref{eq:threshold-H}) and its analogue for
heterogeneous site percolation on a general digraph give a universal
bound for the strong-cluster \emph{uniqueness} transition, below which
a strongly connected infinite cluster cannot be unique.  Such a bound
is continuous for an increasing sequence of subgraphs if the
percolation problem on the limiting digraph has a finite \emph{minimum
  return probability}, the probability that any arc and its inverse
are connected by an open non-backtracking path.  Finite minimum return
probability also guarantees that below this bound, the strong
connectivity decays exponentially with the distance, and the expected
size of a strongly connected cluster scales sublinearly with the
number of vertices in a digraph.  In comparison, the bound
(\ref{eq:threshold-H}) applies only conditionally to the percolation
transition proper, e.g., for a weakly-convergent sequence of
quasi-transitive digraphs of increasing size, where the number of
inequivalent vertex classes remains uniformly bounded.

The remainder of this paper is organized in four sections.  In
Sec.~\ref{sec:defs} we define several matrices associated with
heterogeneous site percolation and introduce other notations.  Our
main results are given in Sec.~\ref{sec:local-chi-bounds} which
contains bounds for finite digraphs, and in Sec.~\ref{sec:infinite}
where infinite digraphs are discussed.  Finally, in
Sec.~\ref{sec:disc} we compare effectiveness of different criteria in
limiting the emergence of a \emph{giant component}, an open cluster
which contains a finite fraction of all vertices in the digraph.

\section{Definitions and notations}
\label{sec:defs}

We consider only simple directed and undirected graphs with no loops
or multiple edges.  A general digraph
$\mathcal{D}=(\mathcal{V}, \mathcal{E})$ is specified by its sets of
vertices (also called sites)
$\mathcal{V}\equiv \mathcal{V}(\mathcal{D})$ and edges
$\mathcal{E}\equiv \mathcal{E}(\mathcal{D})$.  Each edge (bond) is a
pair of vertices, $(u,v)\subseteq\mathcal{E}$ which can be directed,
$u\to v$, or undirected, $u\leftrightarrow v$.  A directed edge
$u\to v$ is also called an arc from $u$ to $v$; an undirected
(symmetric) edge can be represented as a pair of mutually inverted
arcs, $u\leftrightarrow v\equiv\{u\to v,v\to u\}$.  A digraph with no
undirected edges is an oriented graph.  We will denote the set of arcs
in a (di)graph $\mathcal{D}$ as
$\mathcal{A}\equiv \mathcal{A}(\mathcal{D})$.  Each vertex
$v\in\mathcal{V}$ in a digraph $\mathcal{D}$ is characterized by its
in-degree $\id(v)$ and out-degree $\od(v)$, the number of arcs in
$\mathcal{A}(\mathcal{D})$ to and from $v$, respectively.  A digraph
with no directed edges is an undirected graph
$\mathcal{G}=(\mathcal{V},\mathcal{E})$.  For every vertex in an
undirected graph, the degree is the number of bonds that include $v$,
$\deg(v)=\id(v)=\od(v)$.

We say that vertex $u$ is connected to vertex $v$ on a digraph
$\mathcal{D}$, if there is a path from
$u=u_0$ to $v\equiv u_\ell$, 
\begin{equation}
  \label{eq:path}
\mathcal{P}\equiv \{u_0\rightarrow u_1,  u_1\rightarrow
u_2,\ldots, u_{\ell-1} 
  \rightarrow u_\ell\}\subseteq \mathcal{A}(\mathcal{D}).
\end{equation}
The path is called non-backtracking if $u_{i-1}\neq u_{i+1}$,
$0< i<\ell$, and self-avoiding (simple) if $u_i\neq u_j$ for
$0\le i,j\le \ell$.  The length of the path is the number of arcs in
the set, $\ell=|\mathcal{P}|$.  The distance from $u$ to $v$ on
$\mathcal{D}$, $d(u,v)$, is the minimum length of a path from $u$ to
$v$.  We call path (\ref{eq:path}) open if $u_0\neq u_\ell$, and
closed otherwise.  A closed path is a cycle; it can be
non-backtracking or self-avoiding (simple).  Connectivity on an
undirected graph is a symmetric relation: we just say that vertices
$u$ and $v$ are connected (or not).  On a digraph, we say that
vertices $u$ and $v$ are strongly connected iff $u$ is connected to
$v$ and $v$ is connected to $u$; $u$ and $v$ are weakly connected on
${\cal D}$ if they are connected on the undirected graph underlying
${\cal D}$.  A \emph{ray} is a semi-infinite simple path,
characterized as in- or out-going according to the directionality of
the constituent arcs.  A \emph{strong ray} is a strongly connected
union of in- and out-going rays; it has the property that the
intersection between the vertex sets is an infinite set.

A digraph ${\cal D}$ is called \emph{transitive} iff for any two
vertices $u$, $v$ in ${\cal V}\equiv \mathcal{V}(\mathcal{D})$ there
is an automorphism (symmetry) of ${\cal D}$ mapping $u$ onto $v$.
Digraph ${\cal D}$ is called \emph{quasi-transitive} if there is a
finite set of vertices ${\cal V}_0\subset {\cal V}$ such that any
$u\in {\cal V}$ is taken into ${\cal V}_0$ by some automorphism of
${\cal D}$.  We say that any vertex which can be mapped onto a vertex
$u_0\in{\cal V}_0$ is in the equivalence class of $u_0$.  The square
lattice is an example of a transitive graph; a two-dimensional lattice
with $r$ inequivalent vertex classes defines a (planar)
quasi-transitive graph.

A graph $\mathcal{G}'=(\mathcal{V}',\mathcal{E}')$ is called a
covering graph of $\mathcal{G}=(\mathcal{V},\mathcal{E})$ if is there
is a function $f:\mathcal{V}'\to \mathcal{V}$, such that an edge
$(u', v')\in\mathcal{E}'$ is mapped to the edge
$(f(u'), f(v'))\in\mathcal{E}$, with an additional property that $f$
be invertible in the vicinity of each vertex, i.e., for a given vertex
$u'\in \mathcal{V}'$ and an edge $(f(u'), v)\in\mathcal{E}$, there
must be a unique edge $(u', v')\in\mathcal{E}'$ such that $f(v')=v$.
The {\em universal cover\/} $\widetilde{\mathcal{G}}$ of a connected
graph $\mathcal{G}$ is a connected covering graph which has no cycles
(a tree); it is unique, up to isomorphisms.  The universal cover can
be constructed as a graph with the vertex set formed by all distinct
non-backtracking paths from a fixed origin $v_0\in \mathcal{V}$, with
an edge $(\mathcal{P}_1,\mathcal{P}_2)\in\widetilde{\mathcal{E}}$ if
$\mathcal{P}_2=\mathcal{P}_1 \cup u$, $u\in\mathcal{E}$ is a simple
extension of $\mathcal{P}_1$.  Choosing a different origin gives an
isomorphic graph.  The definition of a covering digraph is similar,
except the mapping function $f$ must preserve the directionality of
the edges.  The covering digraph $\widetilde{\mathcal{D}}$ of a
digraph $\mathcal{D}$ can be constructed from that of the underlying
undirected graph by labeling the directionality of the corresponding
edges.

\subsection{Heterogeneous site percolation}
\label{subsec: heterogeneous site percolation}

Consider a connected undirected graph $\mathcal{G}$.  We define
heterogeneous site percolation on $\mathcal{G}$ where each vertex
$v\in\mathcal{V}(\mathcal{G})$ has an associated probability $p_v$,
$0<p_v\le 1$. A vertex is chosen to be open with probability $p_v$,
independent from other vertices.  We are focusing on a subgraph
$\mathcal{G}'\subseteq\mathcal{G}$ induced by all open vertices on
$\mathcal{G}$.  For each vertex $v$, if $v$ is open, let
$\mathcal{C}(v)\subseteq\mathcal{G}'$ be the connected component of
$\mathcal{G}'$ which contains the vertex $v$, otherwise
$\mathcal{C}(v)=\emptyset$.  If $\mathcal{C}(v)$ is infinite, for some
$v$, we say that percolation occurs.  Denote
\begin{equation}
\theta_v\equiv
\theta_v(\mathcal{G},\{p\})=\Pp(|\mathcal{C}(v)|=\infty)
\label{eq:theta-v}
\end{equation}
the probability that $\mathcal{C}(v)$ is
infinite.  Clearly, for any pair of vertices $u$ and $v$,
$\theta_v>0$ iff $\theta_u>0$. 

Similarly, introduce  the connectivity function,
\begin{equation}
  \label{eq:connectivity}
  \tau_{u,v}\equiv
  \tau_{u,v}(\mathcal{G},\{p\})=\Pp\bigl(u\in
  \mathcal{C}(v)\bigr),   
\end{equation}
the probability that vertices $u$ and $v$ are in the same cluster.
For a pair of vertices $u$, $v$ separated by the distance $d(u,v)$, 
$\tau_{u,v}$ can be bounded by the probability that $v$ is in a
cluster of size $d(u,v)+1$.  Thus, in the absence of percolation,
$\tau_{u,v}\to0$ when $d(u,v)\to\infty$.  The reverse is not
necessarily true. 

Yet another measure is  the local susceptibility,
\begin{equation}
\chi_v\equiv
\chi_v(\mathcal{G},\{p\})=\Ep(|\mathcal{C}(v)|), 
\label{eq:chi-v}
\end{equation}
the expected cluster size connected to $v$.  
Equivalently, local susceptibility can be defined as the sum of
probabilities that individual vertices are in the same cluster as $v$,
i.e., in terms of
connectivities, 
\begin{equation}
  \label{eq:chi-v-alt}
  \chi_v=\sum_{u\in{\cal V}}\tau_{v,u}.
\end{equation}
If percolation occurs
(i.e., with probability $\theta_v>0$, $|C_v|=\infty$), then
clearly $\chi_v=\infty$.  The reverse is known to be true in the case
of heterogeneous site percolation on quasi-transitive
graphs\cite{Menshikov-1986,Menshikov-Sidorenko-eng-1987}:
$\chi_v=\infty$ can only happen inside or on the boundary of the
percolating phase.

An important question is the number of infinite clusters on
${\cal G}'$, in particular, whether an infinite cluster is unique.
For infinite quasi-transitive graphs, there are only three
possibilities: (\textbf{a}) almost surely there are no infinite
clusters; (\textbf{b}) there are infinitely many infinite clusters;
and (\textbf{c}) there is only one infinite
cluster\cite{Benjamini-Schramm-1996,%
  Haggstrom-Jonasson-2006,Hofstad-2010}.  This is not necessarily so
for more general graphs.  Notice that when the infinite cluster is
unique, the connectivity function is bounded away from zero,
$\tau_{u,v}\ge\theta_u\theta_v>0$.  In addition, uniqueness of the
infinite cluster implies divergence of the local self-avoiding cycle
(SAC) susceptibility $\chi_\mathrm{SAC}(a)$, the expected number of
distinct simple cycles passing through the arc $a$ on the open
subgraph $\mathcal{D}'$.  In the case of homogeneous percolation on
undirected transitive graphs, such a relation is given by Theorem 3.9
in Ref.~\cite{Lyons-review-2000}, attributed to O.\ Schramm.

When the open-site probabilities are equal for all sites of an
infinite graph, $p_v=p$, $v\in{\cal V}$ (this is homogeneous site
percolation), one defines the critical probabilities $p_c$ and $p_T$,
respectively associated with formation of an infinite cluster and
divergence of site susceptibilities.  There is no percolation,
$\theta_v=0$, for $p<p_c$, but $\theta_v > 0$ for
$p>p_c$. Likewise, $\chi_v$ is finite for $p<p_T$ but not for
$p>p_T$. A third critical probability, $p_u$, is associated with the
number of infinite clusters.  Most generally, we expect
$p_T\le p_c\le p_u$.  For a quasi-transitive graph, one
has\cite{Hofstad-2010}
\begin{equation}
  \label{eq:thresholds}
  0<p_T=p_c\le p_u.
\end{equation}
Here, $p_u$ is the uniqueness threshold, such that there can be only
one infinite cluster for $p>p_u$, whereas for $p<p_u$, the number of
infinite clusters may be zero, or infinite. For a degree-$r$ regular tree
$\mathcal{T}_r$ with $r\ge3$, $p_u=1$, $p_c=1/(r-1)$, while for
hypercubic lattice, $\mathbb{Z}^D$, $p_u=p_c$.

\subsection{Percolation on a general digraph}
\label{subsec: Percolation on a digraph}

There are several notions of connectivity on a digraph, and,
similarly, there are several percolation transitions associated with a
digraph\cite{Restrepo-Ott-Hunt-2008}.  For any given configuration of
open vertices on a digraph $\mathcal{D}$ (which induce the open
digraph $\mathcal{D}'$) we introduce the strongly-connected cluster
which includes $v$,
$\mathcal{C}_{\mathrm{str}}(v)\subseteq \mathcal{D}'$.  Similarly, one
also considers an out-cluster
$\mathcal{C}_\mathrm{out}(v)\subseteq\mathcal{D}'$ and an in-cluster
$\mathcal{C}_\mathrm{in}(v)\subseteq\mathcal{D}'$, formed by all sites
which can be reached from $v$ moving along or opposite the arcs in
$\mathcal{A}(\mathcal{D})$, respectively.  Finally, there is also a
\emph{weakly-connected} cluster $\mathcal{C}_\mathrm{und}(v)$ formed
on the undirected graph $\mathcal{G}'$ underlying $\mathcal{D}'$.  For
each of these cluster types, we introduce the quantities analogous to
those in Eqs.~(\ref{eq:theta-v}), (\ref{eq:chi-v}), and
(\ref{eq:connectivity}), e.g., the probability
$\theta_\mathrm{str}(v)$ that $v$ is in an infinite strongly-connected
cluster, the strongly-connected susceptibility $\chi_\mathrm{str}(v)$,
the two sided (strong) connectivity $\tau_\mathrm{str}({u,v})$ which implies a
path from $u$ to $v$ and one from $v$ to $u$ must both be
open, and the directed connectivity $\tau_{u,v}$ from $u$ to $v$.

\subsection{Emergence of a giant component}

In network theory, a percolation-like transition on a finite graph is
usually associated with the emergence of a \emph{giant component}, an
open cluster which contains a finite fraction of all vertices in a
graph.  The transition is sharp and it is well understood in various
ensembles of random graphs and digraphs, see, e.g.,
Refs.~\cite{Bollobas-giant-2001,Alon-Benjamini-Stacey-2004,%
  Newman-Strogatz-Watts-2001,Chung-Horn-Lu-2009,%
  Bollobas-Borgs-Chayes-Riordan-2010,%
  Benjamini-Boucheron-Lugosi-Rossignol-2012}.

\subsection{Matrices associated with a digraph}
\label{subsec: Matrices associated with a graph}

For any heterogeneous site percolation problem on a digraph with the
adjacency matrix $A$, we associate the following three matrices: \emph{weighted ad\-ja\-cen\-cy matrix} 
\begin{equation}
  \label{eq:Ap}
  [A_p]_{ij}= p_i^{1/2}A_{ij} p_j^{1/2}
 \quad(\text{no summation}),
\end{equation}
weighted line digraph adjacency matrix $L_p$, and
weighted Hashimoto matrix $H_p$. 

\paragraph{Weighted line-digraph adjacency matrix}
\label{subsec: weighted line-digraph adjacency matrix}

For any digraph
$\mathcal{D}$, the line digraph\cite{Harary-Norman-1960}
$\mathcal{L}\equiv \mathcal{L}(\mathcal{D})$ is a digraph whose
vertices correspond to the arcs of the original digraph
$\mathcal{V}(\mathcal{L})=\mathcal{A}(\mathcal{D})$, and it has a
directed edge $(a\to b)\in\mathcal{A}(\mathcal{L})$ between vertices
$a=i\to j$ and $b=j'\to l$ iff $j=j'$ (that is, arcs $a$ and $b$,
taken in this order, form a directed path of length two).  We denote
the corresponding adjacency matrix $L$, and introduce the weighted
matrix $L_p$, where an entry corresponding to the directed edge $(a\to
b)\in\mathcal{A}(\mathcal{L})$ ($a=i\to j$ and $b=j'\to l$ are arcs in
$\mathcal{D}$) has weight $p_j$:
\begin{equation}
  \label{eq:Lp}
  (L_p)_{ab}=p_j\delta_{j,j'}.\quad  \text{(no summation)}
\end{equation}
In the homogeneous case, $p_j=p$ for all sites, and the weighted
matrix has the simple form, $L_p=pL$.
Notice we used the arc set $\mathcal{A}$ to define the line digraph;
the same definition can be used to associate a line digraph
$\mathcal{L}(\mathcal{G})$ with an undirected graph $\mathcal{G}$.

\paragraph{Weighted Hashimoto matrix}
\label{subsec: weighted Hashimoto matrix}

Hashimoto, or non-back\-tracking,
matrix $H$ has originally been defined for counting non-back\-tracking
cycles on graphs\cite{Hashimoto-matrix-1989}.  This matrix is the
adjacency matrix of the \emph{oriented line graph} (OLG)
$\mathcal{H}(\mathcal{D})$ associated with the original
(di)graph\cite{Kotani-Sunada-2000}.  The OLG is defined similarly to
the line digraph, except that the edges corresponding to mutually
inverted pairs of arcs in the original digraph are dropped.  We define
the corresponding weighted matrix by analogy with Eq.~(\ref{eq:Lp}),
\begin{equation}
  \label{eq:Hp}
  (H_p)_{ab}=p_j\delta_{j,j'}(1-\delta_{i,l}),\quad  \text{(no summation)}  
\end{equation}
where $a=i\to j$ and $b=j'\to l$ are arcs in the original digraph.
Again, in the homogeneous case, $p_j=p$ for all sites, we recover the usual
Hashimoto matrix, $H_p=pH$.

Notice that in the case of an infinite digraph, the objects $A_p$,
$L_p$, and $H_p$ are not matrices but operators acting in the
appropriate infinite-dimensional vector spaces.  For a locally-finite
digraph, the action of these operators is uniquely defined,
respectively, by the local rules (\ref{eq:Ap}), (\ref{eq:Lp}), and
(\ref{eq:Hp}).  For convenience we will nevertheless refer the them as
``matrices'', each time specifying whether the graph is finite or
infinite.

\subsection{Perron-Frobenius theory}
\label{subsec: Perron-Frobenius theory}

Consider a square $n\times n$ matrix $B$ with non-negative matrix
elements, $B_{ij}\ge0$, and not necessarily symmetric.  The spectral
radius $\rho(B)\equiv \max_{i\le n}|\lambda_i(B)|$ and the associated
eigenvectors of such a matrix are analyzed in the Perron-Frobenius
theory of non-negative
matrices\cite{Perron-1907,Frobenius-1912,Meyer-book-2000}.  In
particular, there is always an eigenvalue
$\lambda_\mathrm{max}=\rho(B)$, and the corresponding left and right
eigenvectors $\xi_L$ and $\xi_R$, $\xi_LB=\lambda_\mathrm{max}\xi_L$,
$B\xi_R=\lambda_\mathrm{max}\xi_R$, can be chosen to have non-negative
components, $\xi_{Li}\ge0$, $\xi_{Ri}\ge0$, although in general one
could have $\rho(B)=0$.  Further, in the case where $B$ is strongly
connected (as determined by a digraph with the adjacency matrix given
by non-zero elements of $B$), the spectral radius is strictly
positive, as are the components of $\xi_L$, $\xi_R$.  For such a
positive vector $\xi$, we will consider the \emph{height ratio}
\begin{equation}
  \gamma(\xi)\equiv \max_{ij}\frac{\xi_i}{\xi_j},\quad \gamma(\xi)\ge1.
  \label{eq:ratio}
\end{equation}

\section{Finite graph bounds}
\label{sec:local-chi-bounds}
\subsection{Approach}
The derivation of Eq.~(\ref{eq:threshold-H}) in Ref.\
\cite{Hamilton-Pryadko-PRL-2014} relied on the mapping of the
percolation thresholds between the original graph $\mathcal{G}$ and
its universal cover, a tree $\mathcal{T}$ locally equivalent to
$\mathcal{G}$, $p_c(\mathcal{G})\ge p_c(\mathcal{T})$.  Our approach
in this work is mostly algebraic.  We
consider 
the bound (\ref{eq:threshold-H}) as the convergence radius for the
infinite power
series\cite{
  Karrer-Newman-Zdeborova-arXiv-2014},
\begin{equation}
 M\equiv M(H)\equiv  \sum_{s=1}^{\infty}p^s H^s,
 \label{eq:series-H}
\end{equation}
where a matrix element of $H^s$, $[H^s]_{uv}$, gives the number of
non-backtracking paths starting at site $i$ along the arc
$a\equiv i\to j$, ending at the arc $b$, and visiting $s-1$
intermediate sites.  Thus, the sum $\sum_{b}M_{ab}$ is an upper bound
on the average number of sites which can be reached starting along the
arc $a$.  Respectively, in an infinite graph, the convergence of the
series (\ref{eq:series-H}) implies: with probability one any given
point belongs to a finite cluster.  Unfortunately, this argument does
not limit giant components on a finite graph.  Indeed, matrix $H$ is
non-symmetric, and convergence of the series can be highly non-uniform
in $s$ (see, e.g., Ref.~\cite{Bandtlow-2004}), with the norm of each
term exponentially increasing as $p^s\| H^s\| \sim p^s\|H\|^s$ for
$s<s_0=\mathcal{O}(m)$, and only starting to decrease for
$p<1/\rho(H)$ at $s\ge s_0$.  Thus, formal convergence does not
guarantee the low probability of finding a giant component in a large but
finite graph.

On the other hand, the expansion (\ref{eq:series-H}) gives a
convenient tool to study percolation.  By eliminating the contribution
of backtracking paths, and reducing over-counting, one can get bounds
tighter than what would be possible in a similar approach based on the
adjacency matrix.  Now, the original problem of
percolation on an undirected graph is substituted by the problem of
percolation on a directed graph, the OLG whose adjacency matrix is the
non-backtracking matrix $H$.  The same formalism can be used to bound
percolation on digraphs, a problem of high importance in network
theory\cite{Grassberger-1983,Moore-Newman-2000,
  Sander-epidemics-2002,PastorSatorras-Vespignani-2001,
  Gai-Kapadia-2010,Watts-PNAS-2002,Kempe-Kleinberg-Tardos-2003,%
  Jiang-Miao-Yi-Zhenzhong-Hauptmann-2014}. 

\subsection{Spectral bounds on susceptibilities}
\label{sec:upper-spec-bounds}

In the following, we only construct bounds for the out-cluster
susceptibilities.  The corresponding bounds for in-cluster
susceptibilities, $\chi_\mathrm{in}(v)$, can be obtained by
considering the transposed matrices, $A_p^T$ and $H_p^T$,
respectively.

\begin{theorem}
\label{th:chi-out-bound-A}
Consider heterogeneous site percolation on a finite strongly-connected
digraph $\mathcal{D}$.  Assume that the spectral radius of the
weighted adjacency matrix satisfies $\rho(A_p)<1$, and let $ \xi_R$ be the
corresponding right PF vector.
The out-cluster susceptibility for an arbitrary vertex
$v\in\mathcal{V}(\mathcal{D})$ satisfies:%
\begin{equation}
\chi_\mathrm{out}(v)\le {C_1(\xi_R)\over 1-\rho(A_p)},\;\;
C_1(\xi)\equiv \max_{u\subset\mathcal{V}}\max_{v\subset\mathcal{V}}{p_u^{1/2} \xi_u\over
  \xi_v/p_v^{1/2}}\le \gamma(\xi).
\label{eq:chi-out-bound}
\end{equation}  
\end{theorem}

\begin{proof}
Consider the
  alternative definition (\ref{eq:chi-v-alt}) of the susceptibility
  $\chi_v$ as a sum of connectivities.  Any given site $u$ is in the
  out-cluster of $v$ iff there is an open path leading from $v$ to
  $u$.  The corresponding probability can be estimated using the union
  bound, with sum of probabilities over self-avoiding paths
  upper-bounded by matrix elements of powers of the matrix $A_p$.
  Namely, the upper bound for the susceptibility
  $\chi_\mathrm{out}(v)$ reads:
\begin{equation}
\label{eq:out-bound-sum}
\chi_\mathrm{out}(v)\le {p_v^{1/2}}\sum_{s=0}^\infty
\sum_{u\in\mathcal{V}(\mathcal{D})}
[(A_p)^s]_{vu}{p_u^{1/2}}. 
\end{equation}
For each $u$, we replace $p_u^{1/2}\le\xi_{Ru}/\min_i(
\xi_{Ri}/p_i^{1/2})$, which reduces powers of $A_p$ to those of $\rho(A_p)$,
thus
\begin{equation}
\label{eq:2}
\chi_\mathrm{out}(v)\le 
[1-\rho(A_p)]^{-1} {p_v^{1/2}\xi_{Rv}\over \min_i \xi_{Ri}/p_i^{1/2}}. 
\end{equation}
The uniform bound (\ref{eq:chi-out-bound}) is obtained by 
maximizing over $v$.
\end{proof}

In the case of an undirected graph, or a digraph with some undirected
edges, we can try to construct a better bound by considering only
non-backtracking paths, with the corresponding probabilities counted
using the weighted Hashimoto matrix $H_p$.  The argument is simplest
when the OLG of the original graph is also
strongly connected.  We have
\begin{theorem}
\label{th:chi-out-bound-H} 
Consider heterogeneous site percolation on a finite digraph
$\mathcal{D}$ with a strongly connected OLG.  Assume that the spectral radius of the
weighted Hashimoto matrix satisfies $\rho(H_p)<1$, and let $ \eta_R$ be the
corresponding right PF vector.
  The local out cluster susceptibility
satisfies:%
\begin{equation}
\label{eq:chi-bound-H}
\chi_\mathrm{out}(v)\le p_\mathrm{max}+{C_2(\eta_R)\over 1-\rho(H_p)},\;\,
C_2(\eta)\equiv{\max_{u}p_u\sum_i \eta_{u\to i}\over \min_{b\equiv
    u'\to j}\eta_{b}/p_j}.
\end{equation}
\end{theorem}
The proof is analogous to that of Theorem~\ref{th:chi-out-bound-A},
except that the sum (\ref{eq:out-bound-sum}) is replaced by a similar
sum in terms of the weighted Hashimoto matrix, with an additional
summation over all arcs leaving a chosen vertex $v$.  We see that 
$C_2(\eta)\le \od_\mathrm{max}\gamma(\eta)$, where $\od_\mathrm{max}$
is the maximum out degree of a vertex on $\mathcal{D}$.

One possible set of sufficient conditions for an OLG  to be strongly
connected  is given by the following Lemma:
\begin{lemma}
\label{th:olg-strongly-connected} 
Consider a strongly-connected digraph $\mathcal{D}$.  The
corresponding OLG $\mathcal{H}(\mathcal{D})$ is also strongly
connected if either of the following is true: (\textbf{a})
$\mathcal{D}$ has no undirected edges, or (\textbf{b}) ${\cal D}$
remains strongly connected after any undirected edge
$i\leftrightarrow j\in\mathcal{E}(\mathcal{D})$ is replaced by either
of the two arcs, $a\equiv i\to j$ or $\bar a\equiv j\to i$.
\end{lemma}
\begin{proof}
  To connect the arcs $a=i\to j$ and $b=u\to v$ from
  $\mathcal{A}(\mathcal{D})$, take a directed path ${\cal P}=\lbrace  j\to
  j_1, j_1\to j_2,\ldots,j_{m-1}\to u\}$ from $j$ to $u$.  If
  $j_1=i$ [in the case the inverse of $a$ is also in the arc set,
  $\bar a=j\to i\in \mathcal{A}(\mathcal{D})$], replace the first step
  by a directed path from $j$ to $i$ which does not include the arc
  $j\to i$.  If needed, do the same at the other end, and remove any
  backtracking portions in the resulting directed path.
\end{proof}
In the special case of a symmetric digraph $\mathcal{D}$
corresponding to an undirected graph $\mathcal{G}$, the condition
(\textbf{b}) in Lemma \ref{th:olg-strongly-connected} is equivalent to
$\mathcal{G}$ having no leaves\cite{Hamilton-Pryadko-PRL-2014}
(degree-one vertices).

For completeness, we also establish the relation between
the spectral radii of the matrices $A_p$ and $H_p$ (this is an
extended version of Theorem 1 from Ref.\
\cite{Hamilton-Pryadko-PRL-2014}; 
see also Ref.~\cite{Karrer-Newman-Zdeborova-PRL-2014}):
\begin{statement}
  \label{th:spectral-radii-relation}
  (\textbf{a}) The spectral radii of the matrices
  (\ref{eq:Ap}), (\ref{eq:Lp}), and (\ref{eq:Hp}) corresponding to
  heterogeneous site percolation on a finite strongly-connected digraph
  $\mathcal{D}$ satisfy $\rho(A_p)=\rho(L_p)\ge \rho(H_p)$.
  (\textbf{b})~If $\mathcal{D}$ has no undirected edges, all three
  spectral radii are equal.  (\textbf{c})~If the OLG of a finite digraph
  $\mathcal{D}$ is strongly connected and $\rho(H_p)=\rho(A_p)$, then
  $\mathcal{D}$ has no undirected edges.
\end{statement}
\begin{proof} 
  Consider a right PF eigenvector of the matrix $H_p$ with
  non-negative components, $\eta_{a}\ge0$, where $a\equiv i\to j$ goes
  over all arcs in the digraph $\mathcal{D}$.  Such an eigenvector can
  always be constructed, whether or not the corresponding graph, OLG
  of $\mathcal{D}$, is strongly
  connected\cite{Perron-1907,Frobenius-1912,Meyer-book-2000}.
  According to the definition~(\ref{eq:Hp}), the corresponding
  equations are%
  \begin{equation}
  \lambda\eta_{i\to j}=p_j\sum_{l:j\to
    l\in\mathcal{A}(\mathcal{D})}\eta_{j\to l}-p_j\eta_{j\to
    i},\label{eq:EV}
\end{equation}
where $\lambda\equiv\rho(H_p)\ge0$. The last term in Eq.~(\ref{eq:EV})
removes the backtracking contribution; it is present iff the edge
$(i,j)\in\mathcal{E}(\mathcal{D})$ is an undirected edge
$(i\leftrightarrow j)$.  Denote $x_j\equiv \sqrt{p_j}\sum_{l:j\to
  l}\eta_{j\to l}$ and $y_i=\sqrt{p_i}\sum_{j:i\leftrightarrow
  j}p_j\eta_{j\to i}\ge0$, where in the latter case we are only
summing over the neighbors $j$ connected to $i$ by undirected edges;
$y_i=0$ if there are no undirected edges incident on $i$.  In
Eq.~(\ref{eq:EV}), fix $i$, sum over $j$, and multiply by
$\sqrt{p_i}$; this gives%
\begin{equation}
  \label{eq:EV-Ap}
  \lambda  x_i=\sum_{j}\sqrt{p_i}A_{ij}\sqrt{p_j}\,x_j-y_i,
\end{equation}
or just $ \lambda x=A_px-y$.  Denote $\xi$ a left PF vector of $A_p$,
such that $\xi^T A_p=\rho(A_p)\xi^T$.  By assumption, the original
digraph is strongly connected, and all $p_i>0$; therefore vector $\xi$
is unique, up to a normalization factor, which can be chosen to make
components of $\xi$ all positive.  Multiply the derived Eq.\
(\ref{eq:EV-Ap}) by $\xi_i$ and sum over
$i\in\mathcal{V}(\mathcal{D})$; this gives
$$
\lambda\, \xi^T x=\xi^T A_p x-\xi^T y=\rho(A_p)\xi^T x-\xi^T y\le
\rho(A_p)\xi^T x.
$$
Notice that $\xi_i>0$ while $x_i\ge0$ and vector $x\neq0$; thus
$\xi^T x>0$, which proves $\rho(A_p)\ge \lambda\equiv \rho(H_p)$.  The
same argument for the line digraph adjacency matrix $L_p$ produces
identical equations with $y\to0$; this gives\cite{Harary-Norman-1960}
$\rho(L_p)=\rho(A_p)$ and completes the proof of (\textbf{a}).
Similarly, the second term in the RHS of Eq.~(\ref{eq:EV-Ap}) is
absent in the case where $\mathcal{D}$ is an oriented graph with no
undirected edges (in this case any path is automatically
non-backtracking); this proves (\textbf{b}).  Further, if the OLG
${\cal H}$ is strongly connected, we have $\eta_a>0$, the components of the
vector $x$ are all positive, $x_i>0$.  In this case,
$\rho(A_p)>\rho(H_p)$ unless $\mathcal{D}$ has no undirected edges,
which proves (\textbf{c}).
\end{proof}

Since Theorem \ref{th:chi-out-bound-H} was derived by counting only
non-back\-tracking paths, one could expect the corresponding bound to
be tighter whenever some undirected edges are present in
$\mathcal{E}(\mathcal{D})$.  Indeed, in this case
$\rho(H_p)< \rho(A_p)$, and, assuming $\rho(A_p)<1$, the denominator
in Eq.~(\ref{eq:chi-bound-H}) 
is larger than that in Eq.~(\ref{eq:chi-out-bound}).  However, for
sufficiently small $\rho(A_p)$, the relation between the two bounds
is determined by the coefficients $C_1(\xi_R)$ and $C_2(\eta_R)$, and
these depend strongly on the graph.  In particular, in the case of
homogeneous percolation with on-site probability $p$ on a $d$-regular
undirected graph, $C_1(\xi_R)=p$ while $C_2(\eta_R)=d p^2$, so that,
indeed, the second bound is tighter.  On the other hand, $C_2(\eta)$
can be infinite when the OLG is not strongly connected---e.g., when
the graph contains a degree-one vertex.  
[Notice that the
structure of the coefficients $C_1(\xi)$ and $C_2(\eta)$ ensure that
they remain bounded when one of the on-site probabilities $p_v$
becomes very small, as long as the graph remains well connected after
the removal of the vertex $v$ and incident edges.]

Finally, in the case of a strongly-connected oriented graph with no
undirected edges, where Statement \ref{th:spectral-radii-relation}
gives $\lambda\equiv\rho(A_p)=\rho(H_p)$, one  gets
$C_2(\eta_R)=\lambda C_1(\xi_R)$, so that the bounds in Theorems
\ref{th:chi-out-bound-A} and \ref{th:chi-out-bound-H} become almost
identical, as expected.

In the following subsection, we give several bounds on vertex
connectivities.  An exponential decay of connectivity implies a
sublinear scaling of the expected size of an open cluster with
$n\equiv |\mathcal{V}|$, and, by the Markov's inequality, an
asymptotically zero probability of a giant open component which
contains $c n$ or more vertices, with any $c>0$.  This can be seen
from the following
\begin{lemma}
\label{th:exp-decay-susceptibility}
  Consider heterogeneous percolation on a digraph $\mathcal{D}$ with
  $n=|\mathcal{V}(\mathcal{D})|$ vertices.  For a chosen site $v$,
  assume that the directed connectivity decays exponentially with the
  distance, i.e., there are some constants $C>0$ and $\rho<1$ such that
  $\tau_{v,u}\le C \rho^{d(v,u)}$, while the number of sites at
  distance $m$ from $v$ is exponentially bounded,
  $|\{u\in\mathcal{V}:d(v,u)=m\}|\le C' \Delta^m$, for some $C'>0$ and
  $\Delta>1$.  Then, for any $q$ such that $x\equiv \Delta\,\rho^q<1$, 
\begin{equation}
  \label{eq:susceptibility-bound-power}
  \chi_\mathrm{out}(v)\le n^{1-1/q} {C (C')^{1/q}\over (1-x)^{1/q}}.
\end{equation}

\end{lemma}
\begin{proof}
  Consider Eq.~(\ref{eq:chi-v-alt}) as a $1$-norm of an $n$-component
  vector of connectivities.  Eq.~(\ref{eq:susceptibility-bound-power})
  follows immediately from the H\"older's inequality,
  $\| a\|_1\le \|a\|_q n^{1-1/q}$, where the $q$-norm is upper bounded
  in terms of an infinite geometrical series with the common ratio
  $x$.
\end{proof}
Note: for a digraph whose underlying graph has a maximum degree
$d_\mathrm{max}$, $\Delta\le d_\mathrm{max}-1$.

\subsection{Spectral  bounds on connectivity and SAC susceptibility}
\label{sec:upper-spec-conn-bounds}

It is easy to check that whenever the height ratio of the
corresponding PF vector is bounded, all connectivity functions decay
exponentially with the distance if $\rho(A_p)<1$.  However, since only
two vertices are involved, we can get a more general statement:
\begin{theorem}
  \label{th:dir-connectivity-bound-A}
  Consider heterogeneous site percolation on a finite
  digraph $\mathcal{D}$ characterized by the weighted adjacency matrix
  ${A}_p$ with the spectral radius $\rho\equiv \rho(A_p)<1$.
  Then, for any pair of vertices, the directed connectivity either from $u$
  to $v$ or from $v$ to $u$ is exponentially bounded,
  \begin{equation}
    \label{eq:dir-connectivity-bound}
\tau_{u,v}\le (1- \rho)^{-1}\rho^{d(u,v)}, \;\,\text{or}\;\,
\tau_{v,u}\le (1- \rho)^{-1}\rho^{d(v,u)}, 
 \end{equation}
  where $d(u,v)$ is the directed distance from $u$ to $v$. 
\end{theorem}
\begin{proof}
  The statement of the theorem follows immediately if $u$ or $v$ are
  not strongly connected, as in this case either or both directed
  connectivities are zero.  When $u$ and $v$ are strongly connected,
  the corresponding components of the right PF vector of $A_p$ are
  both non-zero, $\xi_u>0$, $\xi_v>0$.  The argument similar
  to that in the proof of Theorem~\ref{th:chi-out-bound-A} gives the
  bound 
  \begin{equation}
    \label{eq:10}
    \tau_{u,v}\le {\xi_u\over \xi_v} {\rho^{d(u,v)}\over 1-\rho}.
  \end{equation}
 The pair with smaller prefactor proves (\ref{eq:dir-connectivity-bound}).
\end{proof}

Theorem \ref{th:dir-connectivity-bound-A} can be interpreted as a
bound on the strong connectivity, and an associated bound on
strong-cluster susceptibility from  Lemma
\ref{th:exp-decay-susceptibility}:
\begin{corollary}
\label{th:undir-connectivity-bound}
Consider heterogeneous site percolation on a finite digraph
$\mathcal{D}$ with $n\equiv |{\cal V}({\cal D})|$ vertices,
characterized by the weighted adjacency matrix 
${A}_p$ with the spectral radius $\rho\equiv \rho(A_p)<1$.
Then,  for any pair of vertices, the probability that they are in the
same strongly connected cluster is exponentially bounded,
  \begin{equation}
    \label{eq:str-connectivity-bound}
    \tau_\mathrm{str}(u,v)\le (1- \rho)^{-1}\rho^{\min[d(u,v),\,d(v,u)]}.
  \end{equation}
In addition, if the  underlying graph of ${\cal D}$ has a maximum degree
$d_\mathrm{max}$,  
strong-cluster susceptibility satisfies, 
\begin{equation}
  \label{eq:str-cluster-bound-A}
  \chi_\mathrm{str}(v)\le C\,n^{1-1/q} ,\quad C={1\over 1-\rho} \left({2\over
      1-x}\right)^{1/q}, 
\end{equation}
for some $q>0$ such that $x\equiv (d_\mathrm{max}-1)\rho^q<1$. 
\end{corollary}

We now would like to prove a tighter version of  Theorem
\ref{th:dir-connectivity-bound-A} which involves non-backtracking
paths and the weighted Ha\-shi\-moto matrix $H_p$, see Eq.~(\ref{eq:Hp}).  To
this end, we first quantify the strong  connectedness of the OLG of
the digraph $\mathcal{D}$: 
\begin{lemma} 
  \label{th:lemma-local-str-connectivity}
  Let $\eta$ be the right PF eigenvector of the matrix $H_p$
  corresponding to the positive eigenvalue
  $\lambda\equiv \rho(H_p)>0$.  Consider a pair of mutually inverted
  arcs $a \equiv i\to j$ and $\bar a\equiv j\to i$,
  $\{a,\bar a\}\subset{\cal A}({\cal D})$, connected by a length-$\ell$ simple
  path,
  \begin{equation}
    \mathcal{P}\equiv \{j_0\to j_1,j_1\to
    j_2,\ldots , j_{\ell-1}\to j_{\ell},j_\ell\to j_0
    \},\label{eq:path-cyc}   
  \end{equation}
  where $j_0\equiv i$, and $j_1=j_\ell\equiv j$; we assume
  $\mathcal{P}\subseteq\mathcal{A}(\mathcal{D})$.  Denote
  \begin{equation}
  P\equiv P(\mathcal{P})\equiv
  \lambda^{-\ell}\prod_{i=1}^{\ell}p_{j_i}\;\,\,\text{and}\;\,\, 
  x_i\equiv\!\!\!    \sum_{l:i\to l\in\mathcal{A}(\mathcal{D})}\eta_{i\to l}.
  \label{eq:P}  
\end{equation}
Then the following inequalities are true:
\begin{eqnarray}
  \label{eq:xi-reversed-bound} 
    \eta_{i\to j}&\ge & {P}\,\eta_{j\to i},\\
    \eta_{i\to j}&\ge& {x_j p_j \,P\over p_j+ \lambda\, P}.
  \label{eq:xi-bound} 
  \end{eqnarray}
\end{lemma}
\begin{proof}
  First, notice that
  $\eta$ satisfies the eigenvalue equation~(\ref{eq:EV}).  Using it
  repeatedly along $\mathcal{P}$, write
  \begin{equation}
    \label{eq:return-ineq}
    \eta_{j_0\to j_1}\ge {p_{j_1}\over \lambda}\eta_{j_1\to j_2}\ge {p_{j_1}
      p_{j_2}\over \lambda^2}\eta_{j_2\to j_3}\ge \ldots
    \ge {\eta_{j_\ell\to j_0}\over \lambda^\ell}\prod_{v=1}^\ell {p_{j_v}} , 
  \end{equation}
  which gives Eq.~(\ref{eq:xi-reversed-bound}).  Second, notice that
  the sum in the RHS of Eq.~(\ref{eq:EV}) equals $x_j$.   Combine
  Eq.~(\ref{eq:EV}) with a similar equation for the inverted arc
  $j\to i$ which contains $x_i$.  Together, the two equations give (we
  assume $\lambda^2\neq p_i p_j$)
\begin{equation}
  \label{eq:xi}
  \eta_{i\to j}=p_j{\lambda x_j-p_i x_i\over \lambda^2-p_i p_j},\quad 
  \eta_{j\to i}=p_i{\lambda x_i-p_j x_j\over \lambda^2-p_i p_j}.
\end{equation}
Substite in Eq.~(\ref{eq:xi-reversed-bound}) to obtain
\begin{equation}
  \label{eq:x-inequality} 
{p_ix_i\over \lambda^2-p_ip_j}\le {\lambda +Pp_i\over p_j+\lambda P} 
{p_jx_j\over \lambda^2-p_ip_j},
\end{equation}
where the denominator is preserved for its sign.  The lower bound
(\ref{eq:xi-bound}) is obtained by substituting back 
into~(\ref{eq:xi}).
\end{proof}
For every vertex $j$ of a digraph $\mathcal{D}$ with a strongly-connected
OLG, let us define the \emph{minimal return probability},
\begin{equation}
  \label{eq:min-ret-prob}
  P_j\equiv \min\left({p_j\over\lambda},\min_{a\equiv l\to j} \max_{\mathcal{P}:a\to \bar a} P(\mathcal{P})\right),\;\,
  a\in\mathcal{A}(\mathcal{D}), 
\end{equation}
where the minimum is taken over all arcs leading to $j$, 
$\mathcal{P}$ is a non-backtracking path (\ref{eq:path-cyc}) connecting
$a\equiv l\to j$ and its inverse, $\bar a\equiv j\to l$, and
$P(\mathcal{P})$ is defined in Eq.~(\ref{eq:P}).  For $a=i\to j $ such
that the inverted arc does not exist, $\bar
a\not\in\mathcal{A}(\mathcal{D})$, we should use $P=p_j/\lambda$, see
Eq.~(\ref{eq:EV}). Notice that thus defined $P_j$ is a local quantity;
we expect it to be bounded away from zero for (di)graphs with many
short cycles.

We can now prove 
\begin{theorem}
  \label{th:dir-connectivity-bound-H}
  Consider heterogeneous site percolation on a finite digraph
  $\mathcal{D}$ characterized by the weighted Hashimoto matrix $H_p$
  with the spectral radius $\rho\equiv\rho(H_p)$ such that $0<\rho<1$.
  Let the OLG of $\mathcal{D}$ be also strongly connected, with a
  non-zero minimal return probability, $P_\mathrm{min}\equiv
  \min_{j\in\mathcal{V}(\mathcal{D})}P_j>0$.
  Then, for any pair of vertices, the directed connectivity either
  from $i$ to $j$ or from $j$ to $i$ is exponentially bounded,%
  \begin{equation}
    \label{eq:dir-connectivity-bound-H}
    \tau_{i,j}\le {\rho^{d(i,j)-1}\over
      (1-\rho)}(2P_\mathrm{min}^{-1}),\;\,\,\text{or}\;\,\,
    \tau_{j,i}\le   {\rho^{d(j,i)-1}\over
      (1-\rho)}(2P_\mathrm{min}^{-1}). 
  \end{equation}
\end{theorem}
\begin{proof}
  Start with the union bound for the directed connectivity, formulated
  in terms of non-backtracking paths from $i$ to $j$,
  \begin{equation}
    \label{eq:connectivity-H}
    \tau_{i,j}\le p_ip_j
 \sum_{m\ge d(i,j)-1} \sum_{l:a\equiv i\to l\in
      \mathcal{A}}\sum_{l':b\equiv l'\to j\in \mathcal{A}} [H_p^m]_{ab},
  \end{equation}
where we assume $i\neq j$, so that $d(i,j)\ge1$.  Construct
a further upper bound by introducing the factor [cf.~Eq.\ (\ref{eq:xi-bound})]
\begin{equation}
{\eta_b\over \eta^{(j)}_{\rm min}}\ge 1, \quad
\eta_{\rm min}^{(j)}\equiv \min_{l':l'\to j\in\mathcal{A}}\eta_{l'\to
  j}\ge {x_j p_j P_j\over p_j +\rho P_j},\label{eq:eta-min-J}
\end{equation}
 and extending the summation
over $b$ to all arcs.  This replaces the matrix element $[H_p^m]_{ab}$
with $\rho^m $,
\begin{equation}
  \label{eq:connectivity-rhoH}
    \tau_{i,j}\le p_ip_j
 \sum_{m\ge d(i,j)-1} \sum_{l:a\equiv i\to l\in
      \mathcal{A}}\rho^m {\eta_{a}\over \eta_{\rm min}^{(j)} } .
\end{equation}
The summation over arcs $a$ from $i$ gives $x_i$ in the numerator, see
Eq.~(\ref{eq:P}), while the lower bound (\ref{eq:eta-min-J}) gives
$x_j$ in the denominator.  We obtain  
$$
  \tau_{i,j}\le p_ip_j
{x_i\over x_j} \sum_{m\ge d(i,j)-1}\rho^m \left({\rho\over
    p_j}+{1\over P_j}\right).
$$
The uniform connectivity bound (\ref{eq:dir-connectivity-bound-H}) is
obtained after the summation using the inequalities
$p_j/\rho\ge P_j\ge P_\mathrm{min}$, $p_i\le1$.%
\end{proof}

This gives a conditionally stronger version of Corollary
\ref{th:undir-connectivity-bound}:
\begin{corollary}
\label{th:undir-connectivity-bound-H}
Consider heterogeneous site percolation on a finite digraph
$\mathcal{D}$ characterized by the weighted Hashimoto matrix
${H}_p$ with the spectral radius $\rho\equiv \rho(H_p)$,
$0<\rho<1$.  Then, if $\mathcal{D}$ satisfies the conditions of Lemma
\ref{th:lemma-local-str-connectivity}, the probability that a pair of
vertices are in the same strongly-connected open cluster is
exponentially bounded,
  \begin{equation}
    \label{eq:undir-connectivity-bound-H}
    \tau_{\rm str}(u,v)\le 
    {\rho^{{\min[d(u,v),\,d(v,u)]}-1}\over (1- \rho)}(2P_\mathrm{min}^{-1}). 
  \end{equation}
\end{corollary}

Exponential decay of connectivity with the distance also implies a
sublinear scaling of the expected size of a strongly-connected
cluster (Lemma \ref{th:exp-decay-susceptibility}):
\begin{corollary}
  Consider heterogeneous site percolation on a digraph $\mathcal{D}$
  with $n\equiv |\mathcal{V}(\mathcal{D})|$ vertices, characterized by
  the weight\-ed Hashimoto matrix ${H}_p$ with the spectral
  radius $\rho\equiv \rho(H_p)$, $0<\rho<1$.  Let $d_\mathrm{max}$ be
  the maximum degree of the undirected graph underlying $\mathcal{D}$,
  and $q\ge 1$ satisfy the inequality
  $x\equiv (d_\mathrm{max}-1)\rho^q<1$.  Then, if the graph satisfies
  the conditions of Lemma \ref{th:lemma-local-str-connectivity}, we
  have the following bound for strongly-connected cluster
  susceptibility,
  \begin{equation}
    \label{eq:chi-bound-power-law}
    \chi_\mathrm{str}(v)\le n^{1-1/q}{(2P_\mathrm{min}^{-1})\over (1-
      \rho)(1-x)^{1/q}}. 
  \end{equation}
\end{corollary}
This bound guarantees that in a large digraph, a giant strongly
connected component occurs with asymptotically zero probability as
long as the corresponding prefactor remains bounded.

We conclude this subsection with the following universal bound on the
expected number of SACs through a given arc:%
\begin{theorem}
  \label{th:chi-SAC-bound}
  Let $H_p$ be the weighted Hashimoto matrix (\ref{eq:Hp}) for
  heterogeneous site percolation on a finite digraph
  $\mathcal{D}$.  Then, if the
  spectral radius $\rho(H_p)<1$, the SAC susceptibility for any arc
  $a\in\mathcal{A}(\mathcal{D})$ is bounded,
\begin{equation}
  \label{eq:SAC-bound}
  \chi_\mathrm{SAC}(a)\le [1-\rho(H_p)]^{-1}.
\end{equation}
\end{theorem}
\begin{proof}
  We first consider the case of a digraph whose OLG is strongly connected. 
  The expected number of SACs, $N_s(a)$, of length $s$ passing through
  the arc $a$ on $\mathcal{D}'$ is bounded by
  \begin{equation}
    N_s(a)\le
    [H_p^s]_{aa}
    \quad
    \text{(no summation.)}\label{eq:Ns-bound-one}
  \end{equation}
  Let $\eta$ be a right PF vector of $H_p$ with all-positive
  components: $\eta_a> 0$, $a\in\mathcal{A}$.  The matrix elements of
  $H_p$ being non-negative, the RHS of Eq.~(\ref{eq:Ns-bound-one}) can
  be further bounded by
  \begin{equation}
    \label{eq:Ns-bound-two}
    [H_p^s]_{aa}=    \sum_{b\in\mathcal{A}} [H_p^s]_{ab}\delta_{ba}<
    \sum_{b\in \mathcal{A}}[H_p^s]_{ab}{\eta_b\over 
      \eta_a}=[\rho(H_p)]^s,  
  \end{equation}
  where there is no implicit summation over the chosen
  $a\in\mathcal{A}(\mathcal{D})$.  Summation over $s$ gives
  Eq.~(\ref{eq:SAC-bound}).

  For a general digraph ${\cal D}$, not necessarily connected, we
  notice that the bound (\ref{eq:SAC-bound}) is independent of the
  actual values of the components of the PF vector $\eta$.  A general
  finite digraph can be made strongly connected by introducing
  additional vertices and additional edges connecting these vertices
  to different strongly connected components of ${\cal D}$, with
  arbitrarily small probabilities $p_i=\epsilon>0$ for these vertices
  to be open.  The statement of the Theorem is obtained in the limit
  $\epsilon\to0$.
\end{proof}

\subsection{Matrix norm bounds}
\label{sec:gen-bounds}

For  digraphs where the height ratios $\gamma_R$ and $\gamma_L$ may be
large, we give a weaker general bound for the local
 in-/out-cluster susceptibilities:
\begin{theorem}
  \label{th:H-norm-one-bound}  
  Let the induced one-norm of the weighted Hashimoto matrix
  corresponding to heterogeneous site percolation on a digraph
  $\mathcal{D}$ satisfy $\|H_p\|_1<1$.  Then, the in-cluster
  susceptibility for vertex $j$ is bounded,
  \begin{equation}
  \chi_\mathrm{in}(j)\le1+
  \id(j)\,(1-\|H_p\|_1)^{-1},\label{eq:H-norm-one-bound}  
\end{equation}
where $\id(j)$ is the in-degree of $j$.
\end{theorem}
We note that $\|H_p\|_1$ equals to the maximum column weight of $H_p$.
The analogous result for the out-cluster susceptibility is obtained by
considering the transposed matrix $H_p^T$, which gives the maximum row
weight of $H_p$ (also, $\|H^T\|_1=\|H\|_\infty$).  
\begin{proof}[Proof of Theorem \ref{th:H-norm-one-bound}] 
  We start with a version of Eq.~(\ref{eq:out-bound-sum}) for
  in-cluster susceptibility,%
  \begin{equation}
    \label{eq:chi-out-bound-H}
        \chi_\mathrm{in}(j)\le p_j+
        \biggl\|p_i\sum_{s=0}^\infty 
    [(H_p)^s]_{vu}{p_j}\biggr\|_1,
  \end{equation}
  where the first term accounts for the starting point $j$, summation
  over the arcs $v\equiv i\to i'$ and $u\equiv j'\to j$ is assumed,
  and the expression inside the norm is a vector with non-negative
  components labeled by the site index $i$.  Statement of the Theorem
  is obtained with the help of the standard norm expansion, also using
  $p_i\le1$.
\end{proof}
Notice that Theorem \ref{th:H-norm-one-bound} works for finite or
infinite digraphs. This implies that in- and out-cluster site-uniform
percolation threshold in an infinite (di)graph satisfy, respectively,
\begin{eqnarray}
  \label{eq:site-uniform-1norm-bound}
  p_c^\mathrm{(in)}&\ge&
  p_T^\mathrm{(in)}\ge \|H\|_1^{-1},\\
  \label{eq:site-uniform-1norm-out-bound}
  p_c^\mathrm{(out)}&\ge&
  p_T^\mathrm{(out)}\ge \|H^T\|_1^{-1}=\|H\|_\infty^{-1}.
\end{eqnarray}
Thus, Theorem \ref{th:H-norm-one-bound} gives a direct generalization
of the well-known maximum-degree bound\cite{Hammersley-1961},
\begin{equation}\label{eq:max-degree-bound}
p_c\ge (d_\mathrm{max}-1)^{-1},
\end{equation}
to heterogeneous site percolation on a digraph.

Notice that any induced matrix norm satisfies the inequality
$\|H_p\|\ge \rho(H_p)$, so that these bounds are generally weaker than
Eq.~(\ref{eq:threshold-H}).  Nevertheless, Example \ref{ex:two-region}
below shows that the bounds (\ref{eq:site-uniform-1norm-bound}) and
(\ref{eq:site-uniform-1norm-out-bound}) are tight: the finiteness of
the height ratio $\gamma_R$ is an essential condition in Theorem
\ref{th:chi-out-bound-A} and Theorem \ref{th:chi-out-bound-H}.

\begin{example}\label{ex:two-region}
  Consider a family of girth-$L$ random strongly connected oriented
  graphs with $2L^2$ vertices, parametrized by in/out degrees
  $d_1\ge d_2>1$ [see Fig.~\ref{fig:multiring}].  The graphs are
  constructed from $2L$ directed cycles of length $L$, by adding $D_i$
  randomly placed arcs from each site of the $i$\,th cycle to those of
  the cycle $(i+1)$, with the constraint that the in-degrees on the
  latter cycle are all the same (and equal to $D_{i}+1$).  We set
  $D_i=d_1-1$ for $i=0,\ldots,L-1$, and $D_i=d_2-1$ for
  $i=L,\ldots,2L-1$, with the sites of the last cycle connecting to
  those of the first.

  In the degree-regular case $d_1=d_2=d$, the PF vectors of the
  Hashimoto matrix have height ratios $\gamma_L=\gamma_R=1$, which
  implies that the in-/out-cluster susceptibilities for this family of
  digraphs is bounded for $p<\rho(H)^{-1}=1/d$.  
  
  More generally, for this family of digraphs, $\rho(A)-1$ equals the
  geometrical mean of the parameters $D_i$.  In the  case $d_1>d_2$, we get
  $\rho(A)=\rho(H)=1+[(d_1-1)(d_2-1)]^{1/2}$, with the height ratios
  $\gamma_L$ and $\gamma_R$ divergent exponentially with $L$.  In this
  case the  more general weaker bounds
  (\ref{eq:site-uniform-1norm-bound}) and
  (\ref{eq:site-uniform-1norm-out-bound}) apply, with
  $\|H\|_1=\|H\|_\infty=d_1> d_2$. 
  Numerically, for $d_1=3$ and $d_2=2$, we find
  $p_c^\mathrm{(out)}= 0.346\pm0.01$, see Fig.~\ref{fig:numerics}.    
\end{example}
\begin{figure*}[htbp]
  \centering
  \includegraphics[width=1.3\columnwidth]{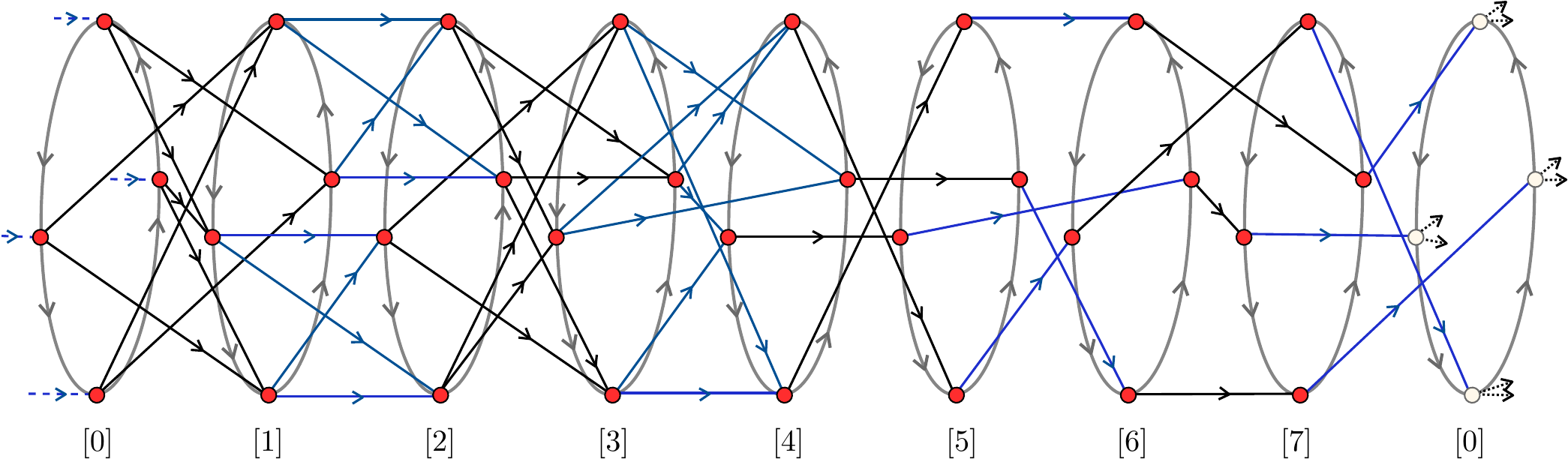}
  \caption{Sample of the two-region digraph from Example
    \ref{ex:two-region}. The digraph is constructed from $2L$ directed
    cycles of length $L$, by adding $D_i$ arcs from each vertex of
    $i$th to randomly selected vertices of the $(i+1)$st cycle, so
    that the in-degrees on cycle $i+1$ all equal $(D_i+1)$.  For $0\le
    i< L$, choose $D_i=d_1-1$, for $L\le i< 2L$ (last ring connects to
    first), $D_i=d_2-1$.  Shown is an actual sample with $L=4$,
    $d_1=3$, $d_2=2$.  For such graphs, at large $L$, in-/out-cluster
    percolation is determined by that in the denser region,
    $p_c^\mathrm{(out)}\gtrsim 1/d_1$ (assuming $d_1>d_2$), see text.}
  \label{fig:multiring}
\end{figure*}

\begin{figure}[htbp]
  \centering
  \includegraphics[width=0.9\columnwidth]{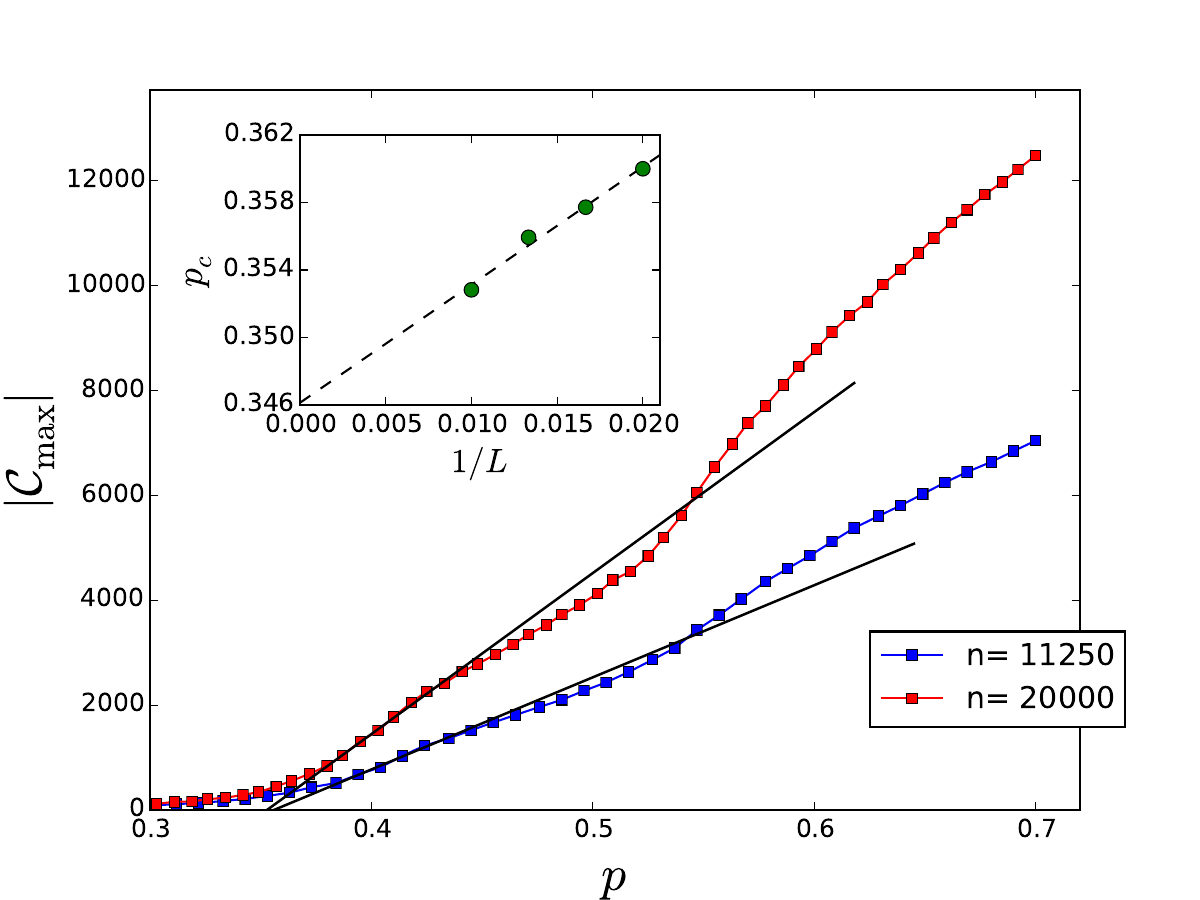}\\
  \caption{Averages of the largest out-cluster size as a function of
    on-site probability $p$ for digraphs in Example
    \ref{ex:two-region}, using $d_1=3$ and $d_2=2$.  Shown are the
    results for two digraphs, with $L=75$ and $L=100$, see
    Fig.~\ref{fig:multiring} for a smaller sample.  Each point is an
    average over 120 realizations of the open vertex configurations
    with given site probability; the corresponding standard errors are
    smaller than the point size.  Lines show the linear fits to the
    numerical data, taken in the regions where the cluster size
    fluctuations are large (not shown).  Inset shows the finite size
    scaling of the corresponding intercepts, giving the out-cluster
    percolation transition at $p_c^{\rm (out)}=0.346\pm0.01$.}
  \label{fig:numerics}
\end{figure}

The bound in Theorem \ref{th:H-norm-one-bound} is local and
independent of the size of the graph.  The following gives bounds for
the in-/out-cluster susceptibilities averaged over all sites:
\begin{theorem}
  \label{th:A-norm-q-bound}
  Let the $q$-norm ($q\ge 1$) of the weighted adjacency matrix
  (\ref{eq:Ap}) for heterogeneous site percolation on a digraph with
  $n$ vertices
  satisfy $\| A_p\|_q<1$.  Then the $q$-th power average of the
  out-cluster susceptibilities satisfies%
  \begin{equation}
    \label{eq:chi-out-bound-q-norm-A}
    \bigl\langle  \chi_\mathrm{out}^q\bigr\rangle
    \equiv {1\over n}\sum_{i=1}^n\chi_\mathrm{out}^q(i)\le \Bigl( 1-\|
      A_p\|_q\Bigr)^{-q}. 
  \end{equation}
\end{theorem}
\begin{proof}
  Write a version of the bound (\ref{eq:out-bound-sum})  as
  \begin{equation}
  \chi_\mathrm{out}(v)\le \sum_{s=0}^\infty [A_p^s ]_{vu}e_u,\label{eq:3} 
\end{equation}
where we introduce the vector with all-one components, $e_u=1$,
$u=1,\ldots,n$.  Consider this expression as an element of a vector
$\vec\chi$ of susceptibilities.  Taking the $q$-norm of this vector
and using the standard norm expansion, we get, in self-evident notations%
  \begin{eqnarray}\nonumber
    \|\vec \chi\|_q&\equiv &
    \left[n\bigl\langle
        \chi_\mathrm{out}^q\bigr\rangle\right]^{1/q}\le 
    \biggl\| \sum_s
      (A_p)^s\, e\biggr\|_q\\
    &\le& (1- \| A_p\|_q)^{-1}\|e\|_q  ={n^{1/q}\over  1- \|
      A_p\|_q}.\label{eq:out-bound-q-norm-one}
  \end{eqnarray}  
Eq.~(\ref{eq:chi-out-bound-q-norm-A}) immediately follows. 
\end{proof}
We note that one can also derive a version of
Eq.~(\ref{eq:chi-out-bound-q-norm-A}) with the non-backtracking
matrix.  However, in the cases where the $q$-norm of a large matrix
can be computed efficiently, $q=1$, $q=2$, and $q=\infty$, the
resulting bounds are superseded by Theorem \ref{th:H-norm-one-bound}:
it is easy to check that\cite{Bordenave-Lelarge-Massoulie-2015}
$\|H_p\|_2=\|H_p\|_1$, and $\|H_p^T\|_2=\|H_p\|_\infty$.

In the special case of heterogeneous site percolation on an undirected
graph $\mathcal{G}$ where the matrix $A_p$ is symmetric, $A_p=A_p^T$,
the two-norm of $A_p$ equals its spectral radius,
$\|A_p\|_2=\rho(A_p)$; in this case we obtain a stronger version of
Corollary \ref{th:undir-connectivity-bound} for undirected graphs:
\begin{corollary}
  \label{th:spectral-radius-generic-chi-bound}
  Let the spectral radius of the weighed adjacency matrix
  (\ref{eq:Ap}) for heterogeneous site percolation on an undirected
  graph with $n$ vertices satisfy $\rho(A_p)<1$.  Then cluster
  susceptibility at any site $i$ is bounded,
  \begin{equation}
    \chi(i)\le {n^{1/2}\over 1-\rho(A_p)}.\label{eq:chi-rho-A-bound}
  \end{equation}
\end{corollary}
This bound guarantees the absence of a giant component in a
sufficiently large graph with $\rho(A_p)<1$.  The scaling with $n$ is
not an artifact of the approximation; this is illustrated by the
following
\begin{example}[Percolation on a rooted tree]
\label{ex:rooted-tree}
Consider an undirected $r$-generation $D$-ary rooted tree with 
 $$n=1+D+\ldots +D^{r}={D^{r+1}-1\over D-1}$$ vertices. Homogeneous
 site percolation 
 with probability $p$ gives the susceptibility at the root of the tree
  $$\chi(0)=p[1+pD+\ldots+(pD)^{r}]=p{(pD)^{r+1}-1\over
    pD-1}.$$ In the limit $r\to \infty$ percolation threshold is
  $p_T=1/D$, same as $p_c$ for the infinite tree $T_d$, $d=D+1$.  On
  the other hand, the spectral radius of the adjacency matrix
  is $$\rho(A)=2\sqrt D\cos\bigl(\pi/(r+1)\bigr) \le 2\sqrt D.$$ 
  For $p=1/\rho(A)>D^{-1}$ (in the case $D>4$), at large $n$,
  $\chi(0)={\cal O}(n^{1/2})$, in agreement with the
  bound~\ref{eq:chi-rho-A-bound}. 
\end{example}

\section{Bounds for infinite digraphs and digraph sequences}
\label{sec:infinite}
\subsection{Bounds for infinite digraphs}
\label{sec:infinite-digraph}

In the majority of graph theory applications, one is interested in
perhaps large but nevertheless finite graphs, considering the
percolation transition as an idealization of a finite-graph crossover
which gets sharper as $n\to\infty$.  It would be conceptually easiest
to define an infinite limit by considering a sequence of increasing
induced subgraphs,
$\mathcal{V}(\mathcal{G}_i)\subset \mathcal{V}(\mathcal{G}_{i+1})$,
such that every vertex in the infinite graph $\mathcal{G}$ is
eventually covered.  However, this is not necessarily a good idea.  In
addition to regular $D$-dimensional lattices (e.g., $\mathbb{Z}^D$),
or graphs which can be embedded in $\mathbb{Z}^D$ with finite
distortions, there are many important graph families which are, in
effect, infinite-dimensional.  Such graphs are characterized by a
non-zero Cheeger constant, 
\begin{equation}
  \label{eq:cheeger}
  h_G=\min_{\mathcal{G}'\subset \mathcal{G}}{|\partial G'|\over
    \min(|\mathcal{G}'|,|\mathcal{G}\setminus\mathcal{G}'|)}, 
\end{equation}
where $|\partial \mathcal{G}'|$ is the size of the boundary, number of
edges connecting the subgraph $\mathcal{G}'$ and its complement,
$\mathcal{G}\setminus\mathcal{G}'$, and, e.g., $|{\cal G}'|$ is the
number of vertices in the subgraph.  An infinite graph with a non-zero
Cheeger constant is called \emph{non-amenable}.  For such a graph, a
boundary of a subgraph may dramatically affect its properties, 
including the nature and the location of the percolation transition in
the infinite-size limit.

An important property of a non-amenable graph $\mathcal{G}$ is that
even if the analog $\xi$ of the PF eigenvector of the graph's
adjacency matrix $A$ can be constructed, $A\xi=\lambda_{\rm max}\xi$
(e.g., in the case of a quasi-transitive graph), $\xi$ cannot be
approximated by a sequence of normalizable vectors.  Moreover, if one
considers $A$ as a bounded operator acting in the Hilbert space
$l^2(\mathcal{V})$ of finite-two-norm vectors in the vertex space
$\mathcal{V}(\mathcal{G})$, the PF eigenvalue $\lambda_\mathrm{max}$
lies outside of the spectrum $\sigma(A)$.  Indeed, if we
take any finite subset $\mathcal{W}\subset \mathcal{V}$, the
corresponding finite-support approximation $\xi'$ of $\xi$ obtained by
setting $\xi_v'=\xi_v$ for $v\in\mathcal{W}$, and zero otherwise,
violates the bonds along the boundary $\partial W$. As a result, the
ratio
\begin{equation}
{\| A \xi'-\lambda_\mathrm{max}\xi'\|_2\over 
  \|\xi'\|_2}
\label{eq:l2-ratio}
\end{equation}
does not become arbitrarily small, no matter how large the support
$|\mathcal{W}|$ is.  
In comparison, the ratio~(\ref{eq:l2-ratio}) can
be made arbitrarily small in an amenable graph with $h_\mathcal{G}=0$.
To distinguish from the usual PF spectral radius (with a
non-normalizable eigenvector), 
we will denote the spectral radius of $A$ treated as an operator in
$l^2(\mathcal{V})$ by $\rltwo(A)\equiv
\sup\{|\lambda|:\lambda\in\sigma(A)\}$, with the spectrum%
\begin{equation}
\sigma(A)\equiv \{\lambda: (\lambda I-A)^{-1} \text{ is not bounded in
}l^2(\mathcal{V})\}.\label{eq:spectrum-A} 
\end{equation}

Regardless of these complications, extending our percolation bounds to
infinite digraphs is conceptually simple.  Indeed, e.g., in the case of
uniform undirected percolation, the key
ob\-servation\cite{Hamilton-Pryadko-PRL-2014} is that percolation
transition $p_c$ on any graph $\mathcal{G}$ cannot be lower than that
on the universal cover $\widetilde{\mathcal{G}}$, $p_c\ge \tilde p_c$.
Universal cover is a tree, and the percolation transition on a locally
finite tree $\mathcal{T}$ can be found in terms of its branching
number\cite{Lyons-1990}, $\tilde p_c=1/\br({\cal T})$.  It turn, the
branching number can be bounded with the \emph{growth} of the
tree\cite{Lyons-1990},
\begin{equation}
  \label{eq:growth}
  \gr{\mathcal{T}}=\inf \Bigl\{\lambda>0: \liminf_{m\to\infty}
  \lambda^{-m}|S_m(i)|=0\Bigr\}, 
\end{equation}
where $|S_m(i)|$ is the number of sites at distance $m$ from the
chosen origin $i$; this definition is independent on the choice of the
initial site $i$.  Notice that when the tree
$\mathcal{T}=\widetilde{\cal G}$ is the universal cover of a graph
$\mathcal{G}$, $|S_m(i)|$ is a number of length-$m$ non-backtracking
paths starting at $i$; this number is the same on a graph and its any
cover and can be counted with the help of the Hashimoto matrix.

To construct a similar bound for the case of heterogeneous site
percolation on a locally-finite digraph, we similarly associate
\emph{growth} with any non-negative matrix,
\begin{equation}
  \label{eq:growth-matrix}
  \gr H=\max_a\inf \Bigl\{\lambda>0: \liminf_{m\to\infty}
  \lambda^{-m}\sum_{b}[H^m]_{ab}=0\Bigr\}.   
\end{equation}
It is a simple exercise to check that on a strongly-connected digraph
the infinum is independent of the chosen arc $a$.  We prove the following
generalization of the bound~(\ref{eq:threshold-H}):
\begin{theorem}
  \label{th:general-percolation-bound}  
  Consider heterogeneous site percolation on a strongly connected
  locally-finite digraph $\mathcal{G}$.  Let on-site probabilities be
  such that the growth of the corresponding weighted Hashimoto matrix
  satisfy $\gr H_p<1$.  Then, with probability one, any out cluster on
  the open subgraph is finite.
\end{theorem}
The proof is a variant of the first half of the proof of Theorem 6.2
in Ref.~\cite{Lyons-1990}.
\begin{proof} Suppose an open out-cluster for a chosen vertex
  $v\in\mathcal{V}$ is infinite.  Then, there should be a simple path
  from $v$ to infinity, entirely constructed from open vertices.  Then
  there is also a directed SAW from some arc $a\in\mathcal{A}$ leaving
  $v$ to infinity, which implies that for any $m>0$, the number of
  points reachable from $a$ in $m$ steps should be at least one.
  Given the probability $\theta_{\rm out}(v)\ge0$ to have an infinite
  out-cluster starting from $v$, we have the following uniform bound
  for any $m>0$,
  \begin{equation}
    \label{eq:Hp-power-sum}
    p_v d_v\max_{a\equiv v\to u}\sum_{b\in\mathcal{A}} [H_p^m]_{ab}\ge
\theta_\mathrm{out}(v).   
\end{equation}
By assumption $\gr H_p<1$, thus
$\liminf_{m\to\infty}\sum_b [H_p^m]_{ab}=0$ for any arc
$a\in\mathcal{A}$.  That is, there should be an increasing sequence
$m_i$, $i=1,2,\ldots$, such that the corresponding limit is zero.
Then, Eq.~(\ref{eq:Hp-power-sum}) gives $\theta_\mathrm{out}(v)=0$,
QED.
\end{proof}

Similarly, to get a bound on the transition associated with divergent
susceptibility, we need to ensure the convergence of the series
similar to that in Eq.~(\ref{eq:2}).  To this end, we introduce the
\emph{uniformly bounded growth}\cite{Angel-Friedman-Hoory-2015},
\begin{equation}
  \label{eq:bargr}
  \bargr{H}\equiv \inf\Bigl\{\lambda>0:
  \limsup_{m\to\infty}\lambda^{-m}\sup_u\sum_v 
  [H^m]_{uv}=0\Bigr\}.
\end{equation}
In particular, if $H$ is a Hashimoto matrix of a tree,
then\cite{Lyons-1990} $p_T=1/\bargr H$.  More generally, we have
\begin{theorem}
  \label{th:out-cluster-infinite}
  Consider heterogeneous site percolation on a bounded-degree digraph 
  $\mathcal{D}$, characterized by the weighted Hashimoto matrix $H_p$.
  If\, $\bargr H_p<1$, then there exists a constant $C>0$ such that
  for any $v\in{\cal V}({\cal D})$, 
  $\chi_\mathrm{out}(v)\le 1+C\od(v)$.
\end{theorem}

\begin{proof}
  The definition (\ref{eq:bargr}) implies that for any
  $\lambda>\bargr H_p$ and any $\epsilon>0$, there exists $m_0$ such
  that for all $m>m_0$,
  \begin{equation}
   \lambda^{-m} \sup_{a\in{\cal A}({\cal D})}\sum_{b\in{\cal A}({\cal
       D})} [H_p^m]_{ab}<\epsilon .\label{eq:6} 
  \end{equation}
  Take $\epsilon=1$ and $\lambda=(1+\bargr H_p)/2<1$, and define a
  constant 
  \begin{equation}
    A\equiv \sum_{m=0}^{m_0}   \sup_a\sum_b
    [H_p^m]_{ab}. \label{eq:7}  
  \end{equation}
  The supremum is finite for any finite $m$ by the assumption of a
  finite maximum degree.  Now, use our usual bound for the out-cluster
  susceptibility, in terms of all non-backtracking paths from
  $v\in {\cal V}({\cal D})$, counted using the powers of $H_p$.
  Summation over $m$ gives the stated bound with
  $C=A+\lambda^{m_0+1}/(1-\lambda)$.
\end{proof}

In comparison, the connectivity is bounded in terms of the $l^2(\mathcal{A})$
spectral radius $\rltwo(H_p)$:
\begin{theorem}
  \label{th:connectivity-bound-infinite}
  Consider heterogeneous site percolation on a locally-finite digraph
  $\mathcal{D}$ characterized by the weighted Hashimoto matrix $H_p$.
  Then, if $\rho\equiv \rltwo(H_p)<1$, there exists a base $\rho'<1$
  and a constant $C\ge(1-\rho')^{-1}$ such that the directed
  connectivity between arcs $u$ and $v$ decays exponentially with the
  distance $d(u,v)$,
  \begin{equation}
    \label{eq:tau-bound-general}
    \tau_{u,v}\le C (\rho')^{d(u,v)}. 
  \end{equation}
\end{theorem}
\begin{proof}
  Start with our usual bound in terms of the weighted Hashimoto
  matrix, 
  \begin{equation}
    \label{eq:connectivity-bound}
    \tau_{u,v}\le \sum_{m\ge d(u,v)} [H_p^m]_{u,v}; 
  \end{equation}
  the series converges since $\rho<1$ by assumption.  Notice that in
  $l^2(\mathcal{A})$, the spectral radius of $H_p$ can be defined as
  the limit,
  \begin{equation}
  \rho=\lim_{m\to\infty}\|H_p^m\|^{1/m},\label{eq:rho-norm-limit} 
\end{equation}
where 
  $\|H_p^m\|^{1/m}\ge\rho$ for every $m>0$.   Then, for any
  $\epsilon>0$, there exists $m_0$ such that
  $\|H_p^m\|^{1/m}<\rho+\epsilon$ for all $m\ge m_0$.  Choose
  $\epsilon=(1-\rho)/2$, define $\rho'=\rho+\epsilon=(1+\rho)/2$, and
  also the constant $B=\max_{0\le m<m_0} \|H_p^m\|/(\rho')^m$, $B\ge1$.   The
  statement of the theorem is satisfied with
  $C=(1-\rho')^{-1}B$. 
\end{proof}

 The spectral radius
$\rltwo(H_p)$ satisfies the following  bounds:
\begin{statement}
  \label{th:Hp-bounds}
Consider heterogeneous site percolation on a locally-finite digraph
$\mathcal{D}$ characterized by the weighted Hashimoto matrix $H_p$,
and the corresponding problem on the universal cover $\widetilde{\cal
  D}$ characterized by the matrix $\tilde{H}_p$.  The following bounds
are true:
\begin{equation}
  \label{eq:Hp-bounds}
  \rltwo(\tilde H_p)\le   \rltwo( H_p)\le \bargr(H_p). 
\end{equation}
\end{statement}
\begin{proof}
  For any $\lambda\neq 0$ and $u\in \mathcal{A}$, define the vectors
  $\xi^{(m)}(u)$ with components 
  $\xi_v^{(m)}(u)\equiv\lambda^{-m}[H_p^m]_{uv}$ and
  $\xi(u)\equiv \sum_{m\ge0} \xi^{(m)}(u)$.  The parameter
  $\lambda\neq0$ is outside of the spectrum of $H_p$ iff the series
  $\sum_{m\ge0}H_p^m/\lambda^m$ define a bounded operator on
  $l^2(\mathcal{A})$ [cf.\ Eq.\ (\ref{eq:spectrum-A})].  Equivalently,
  since individual arcs form a basis of $l^2(\mathcal{A})$, for any
  $u\in\mathcal{A}$, the vector $\xi(u)$ should have a finite norm
  $\|\xi(u)\|_2$.  Since the matrix elements $(H_p)_{uv}\ge0$, the
  sum in Eq.~(\ref{eq:bargr}) is the 1-norm of $\xi^{(m)}(u)$.  On the
  other hand, for any  $\lambda>\bargr H_p$, we can write 
  \begin{equation}
    \label{eq:5}
    \| \xi(u)\|_2\le \sum_{m\ge 0}\|\xi^{(m)}(u)\|_2\le 
    \sum_{m\ge 0}\|\xi^{(m)}(u)\|_1, 
  \end{equation}
  where the rightmost series is convergent since it is asymptotically
  majored by the sum of $(\bargr H_p/\lambda)^m$; this proves the
  upper bound.  To prove the lower bound, consider a similarly defined
  vector $\tilde \xi(u)$ on the universal cover; components $\xi_v(u)$
  are the sums of the non-negative components of $\tilde\xi(u)$
  corresponding to different non-backtracking paths from $u$ to $v$;
  we thus have $\|\tilde \xi(u)\|_2\le \|\xi(u)\|_2$.  Thus, a point
  outside the spectrum $\sigma(H_p)$ is also outside the spectrum
  $\sigma(\tilde{ H}_p)$.
\end{proof}

Further, in the case of homogeneous percolation, the spectral radius 
$\rltwo(\tilde H)=[\br \widetilde{\cal D}]^{1/2}$ is exactly the
tree's point spectral radius\cite{Lyons-1990}, where the branching
number $\br \widetilde {\cal D}\le \gr H\le \bargr H$.  Also, for any
quasi-transitive digraph $\mathcal{D}$, we have
$\br \widetilde {\cal D}=\bargr {H}$; the corresponding value can be
found as the spectral radius $\rho(H')$ for a finite graph.

The following Theorem is a generalization of Theorem 3.9 from
Ref.~\cite{Lyons-review-2000}, attributed to O.\ Schramm.
\begin{theorem}
  \label{th:sac-divergent}
  Consider heterogeneous site percolation on an infinite locally
  finite digraph $\mathcal{D}$, with site probabilities bounded from
  above, $p_v\le p_\mathrm{max}<1$, $v\in\mathcal{V}(\mathcal{D})$.
  Then, if strong percolation occurs, and a strongly connected
  infinite open cluster is unique with probability one, the SAC
  susceptibility is unbounded,
  \begin{equation}
  \sup_{a\in\mathcal{A}}\chi_{\rm SAC}(a)=\infty.\label{eq:sac-divergent}
\end{equation}
\end{theorem}
\begin{proof}
  First, uniqueness of the infinite strongly-connected open cluster
  $K\subseteq\mathcal{D}'$ implies that, with probability one, $K$ is
  one-ended: it cannot be separated into two or more strongly
  connected infinite components by removing any finite set of
  vertices.  Indeed, otherwise, we would have a non-unique infinite
  cluster with a finite probability, which contradicts the assumption.
  Second, with probability one, $K$ contains two disjoint strong rays.
  Indeed, let us assume that not to be the case.  Then, according to
  Menger's theorem, for any $v\in\mathcal{V}(\mathcal{D}')$, the open
  subgraph would have an infinite number of single-vertex cut sets
  separating $v$ from infinity or infinity from $v$.  This would imply
  $\theta_\mathrm{str}(v)=0$, counter to the assumption.  The
  one-endedness of $K$ implies that outside of any finite ball, the
  two strong rays must remain strongly connected with each other.
  This means that with probability one, for some $a\in\mathcal{A}$,
  the open strongly-connected cluster $K$ contains an infinite number
  of simple cycles passing through $a$, which implies
  Eq.~(\ref{eq:sac-divergent}).
\end{proof}

Notice that an upper bound on probabilities $p_v$ is an essential
condition.  This eliminates the case of a digraph with a
ray whose vertices all have $p_v=1$.

\subsection{Bounds for sequences of finite digraphs}
In Sections \ref{sec:upper-spec-bounds} and
\ref{sec:upper-spec-conn-bounds}, we constructed several upper bounds
for susceptibilities and connectivity in heterogeneous percolation on
finite digraphs, formulated in terms of the spectral radius of the
corresponding weighted Hashimoto matrix $H_p$ and the associated PF
vectors $\eta_L$, $\eta_R$.  In contrast, bounds in Section
\ref{sec:infinite-digraph} are formulated directly on infinite
digraphs.   We would like to see the correspondence between these
bounds for weakly convergent digraph sequences. 

Given an infinite graph $\mathcal{G}(\mathcal{V},\mathcal{E})$, we say
that a sequence of graphs $\mathcal{G}^{(t)}$, $t=1,2,\ldots$,
(weakly) converges to $\mathcal{G}$ near the origin
$v_0\in\mathcal{V}$, if for any $R>0$ there is $t_0$ such that the
radius-$R$ vicinity of $v_0$ on $\mathcal{G}$ for every $t\ge t_0$ is
isomorphic to a subgraph of $\mathcal{G}^{(t)}$.  For heterogeneous
site percolation, we require $p_v$ to match on corresponding sites.
For digraphs, we also require the bond directions to match (while
using undirected distance to define the ball).  In the following, when
discussing a sequence of digraphs, objects referring to the digraph
$\mathcal{D}^{(t)}$ are denoted with the corresponding superscript,
e.g., the weighted Hashimoto matrix $H_p^{(t)}$ and its right PF
vector $\eta^{(t)}\in l^2(\mathcal{A}^{(t)})$.  We will also use the
same notation for the corresponding vector mapped to $\mathcal{A}$
under the isomorphism map, $\eta^{(t)}\in l^2(\mathcal{A})$, dropping
any arcs not in $\mathcal{A}$, and adding zeros for arcs not in
$\mathcal{A}^{(t)}$.

We first compare the bounds for the transition associated with
emergence of infinite cluster and divergent out-cluster
susceptibilities.  The bound in Theorem \ref{th:chi-out-bound-H} is
formulated in terms of the spectral radius $\rho(H_p)$ of the
Hashimoto matrix and a prefactor $C_2(\eta_R)$, while bounds in
Theorems \ref{th:general-percolation-bound} and
\ref{th:out-cluster-infinite} are formulated in terms of the growth
$\gr H_p$ and uniformly-bounded growth $\bargr H_p$.  A sufficient
condition for continuity between these bounds is given by the
following:
\begin{statement}
  \label{th:rho-limit-gamma-bounded}
  Consider heterogeneous site percolation on an infinite digraph
  $\mathcal{D}$, characterized by the weighted Hashimoto matrix $H_p$,
  and a sequence of finite digraphs $\mathcal{D}^{(t)}$ with strongly
  connected OLGs, converging to $\mathcal{D}$ around some
  origin. Then, if the right PF vectors of $H_p^{(t)}$ have uniformly
  bounded height ratios, $\gamma(\eta_R^{(t)})\le M$, the
  following limit exists, and
  \begin{equation}
    \label{eq:rho-limit-gamma-bounded-equation}
    \lim_{t\to\infty}\rho(H_p^{(t)})= \gr(H_p)= \bargr(H_p).    
  \end{equation}
\end{statement}
\begin{proof}
  For each $\mathcal{D}^{(t)}$, the corresponding universal cover
  $\widetilde{\cal D}^{(t)}$ is a quasi-transitive tree; a lift
  $\tilde\eta_R^{(t)}$ of the PF vector $\eta_R^{(t)}$ is the
  eigenvector of $H_p^{(t)}$ with all positive components.  For a
  given $u\in\mathcal{A}$ and $t$ large enough, the conditions
  guarantee that the radius-$m$ vicinity of $u$ on $\mathcal{D}$ is
  entirely within the subgraph of $\mathcal{D}^{(t)}$ isomorphic with
  that of $\mathcal{D}$; the same is true for the universal covers.
  We can therefore construct the upper and lower bounds on the sum in
  Eq.~(\ref{eq:bargr}) in terms of the right PF vectors $\eta^{(t)}$,
\begin{equation}
  \label{eq:Hm-bound}
  {1\over \gamma(\eta_R^{(t)})}  [\rho(H_p^{(t)})]^m\le  \sum_v [H_p^m]_{uv}
  \le \gamma(\eta_R^{(t)}) 
  [\rho(H_p^{(t)})]^m,
\end{equation}
or, using the assumed uniform bound on the height ratios,
\begin{equation}
  \label{eq:Hm-bound-uniform}
  {1\over M}  [\rho(H_p^{(t)})]^m\le  \sum_v [H_p^m]_{uv}
  \le  M    [\rho(H_p^{(t)})]^m.
\end{equation}
Now, let us choose a subsequence with spectral radii converging to
$\liminf_{t\to\infty}\rho(H_p^{(t)})$.  Using only these digraphs in the
upper bound~(\ref{eq:Hm-bound-uniform}), the definition
(\ref{eq:bargr}) implies
$$\bargr H_p\le \liminf_{t\to\infty}\rho(H_p^{(t)}).$$  The same
calculation can be repeated for the lower bound, with a subsequence of
graphs whose spectral radii converge to the corresponding superior
limit; we get  
  \begin{equation}
    \label{eq:rho-limit-gamma-bounded}
    \limsup_{t\to\infty}\rho(H_p^{(t)})\le\gr (H_p)\le \bargr(H_p)\le 
    \liminf_{t\to\infty}\rho(H_p^{(t)}).    
  \end{equation}
  This implies that the limit exists and satisfies
  Eq.~(\ref{eq:rho-limit-gamma-bounded-equation}). 
\end{proof}

Similarly, for a finite digraph, Corollary
(\ref{th:undir-connectivity-bound-H}) gives strong connectivity
 exponentially decaying with the distance for $\rho(H_p)<1$,
as long as the corresponding OLG is locally strongly connected.  If we
consider a sequence of such finite digraphs converging to an infinite
digraph around some origin, we expect exponential decay of strong
connectivity for $ \liminf_{t\to\infty} \rho(H_p^{(t)})<1$, with a
bounded prefactor.  On the other hand, Theorem
\ref{th:connectivity-bound-infinite} gives exponential decay of
directed connectivity with the distance on an infinite digraph with
$\rltwo(H_p)<1$, without an explicit bound on the prefactor.  The
following gives partial correspondence between these results for
undirected graphs:
\begin{statement}
  Consider heterogeneous site percolation on a locally-finite infinite
  connected graph $\mathcal{G}$, characterized by the weighted
  Hashimoto matrix $H_p$, and a sequence of finite graphs
  $\mathcal{G}^{(t)}$ converging to $\mathcal{G}$ around some origin.
  If the ratios of the components of the right PF vectors $\eta^{(t)}$
  of $H_p^{(t)}$ corresponding to each arc and its inverse are
  uniformly bounded by a fixed $M\ge1$,
  \begin{equation}
    \label{eq:gamma-ratio-bound}
    M^{-1}\le {\eta^{(t)}_a\over \eta^{(t)}_{\bar a}}\le M, \text{ then
    } \liminf_{t\to\infty}\rho(H_p^{(t)})\ge \rltwo(H_p).
  \end{equation}  
\end{statement}
\begin{proof}
  Define the parity operator as in
  Ref.~\cite{Bordenave-Lelarge-Massoulie-2015},
  $P_{ab}=\delta_{\bar a, b}$, to connect each arc $a$ with its inverse
  $\bar a$; we have $PH_pP=H_p^T$.  The condition of the theorem
  allows to use the components of $P\eta^{(t)}$ as a lower or an upper
  bound on those of $\eta^{(t)}$.  It also guarantees that the OLGs of
  the graphs $\mathcal{G}^{(t)}$ are strongly connected.  For any
  $m>0$ and $v\in\mathcal{A}$, there exists $t_0$ such that for all
  $t>t_0$, the radius-$m$ vicinity of $v$ will be isomorphic to a
  subgraph of $\mathcal{G}^{(t)}$.  If we introduce the vector $e_v$
  with the only non-zero component (equal to one) at the arc $v$, we
  can write for $t>t_0$,
  \begin{equation}
    \label{eq:two-norm-bound}
    \| H_p^m e_v\|_2^2= e_v^T PH_p^mP H_p^m e_v= e_v^T
    P(H_p^{(t)})^mP (H_p^{(t)})^m e_v ,
  \end{equation}
where we used the same notation for the corresponding vector under the
isomorphism map. 
We can now use the PF vector $\eta^{(t)}$ to construct the upper bound,%
\begin{equation}
  \label{eq:two-norm-bound-two}
      \| H_p^m e_v\|_2^2\le e_v^T
    P(H_p^{(t)})^mP (H_p^{(t)})^m \eta^{(t)}{1\over \eta^{(t)}_v}\le M^2
    [\rho(H_p^{(t)})]^{2m}, 
\end{equation}
where we used the identity $[P\eta^{(t)}]_v\le M\eta^{(t)}_v$ twice.
This shows that for any $v\in\mathcal{A}$ and any
$\lambda>\liminf_{t\to\infty}\rho(H_p^{(t)})$, 
the two-norm of 
$H_p^me_v /\lambda^m$ converges to zero, i.e., the operator
$(\lambda I-H_p)^{-1}$ is bounded in $l^2(\mathcal{A})$, thus
$\rltwo(H_p)\le \liminf_{t\to\infty}\rho(H_p^{(t)})$. 
\end{proof} 
We note that for an increasing sequence of subgraphs of $\mathcal{G}$,
$\rho(H^{(t)}_p)\le \rho_{l^2}(H_p)$ in a non-decreasing bounded
function of $t$.  Also, the condition on the PF vector can be
guaranteed by Lemma~\ref{th:lemma-local-str-connectivity}.  This
implies
\begin{corollary}
  \label{th:subgraph-sequence-local}
  Consider heterogeneous site percolation on a  locally-finite infinite graph
  $\mathcal{G}$ whose OLG is locally strongly connected, i.e., for
  every arc $a\in\mathcal{A}(\mathcal{G})$, there is a
  non-backtracking path of length at most $\ell$ from $a$ to $\bar a$.
  Assume that the site probabilities are bounded from below,
  $p_v>p_\mathrm{min}>0$.  Consider any increasing sequence of finite
  subgraphs
  $\mathcal{G}^{(t)}\subset \mathcal{G}^{(t+1)}\subset \mathcal{G}$,
  convergent to ${\cal G}$ around some origin.  Then,
  \begin{equation}
\lim_{t\to\infty}\rho({H_p^{(t)}})=\rltwo(H_p).\label{eq:rho-limit-two}
\end{equation}

\end{corollary}

Our final result is a bound on the SAC susceptibility for a weakly
convergent sequence of finite digraphs, and an associated bound for
the transition associated with the number of strongly-connected
infinite clusters.  We notice that the SAC susceptibility counts only
finite-length cycles.  This implies, most generally:
\begin{theorem}
  \label{th:SAC-bound-infinite}
  Consider heterogeneous site percolation on an infinite locally finite
  digraph $\mathcal{D}$, characterized by the weighted Hashimoto
  matrix $H_p$, and an increasing sequence of finite subgraphs
  $\mathcal{D}^{(t)}\subset \mathcal{D}^{(t+1)}\subset\mathcal{D}$
  converging to $\mathcal{D}$ around some origin.  Consider
  $\rho_0\equiv \lim_{t\to\infty}\rho(H_p^{(t)})$.  Then, if
  $\rho_0<1$, the SAC susceptibility at any arc $a\in\mathcal{A}$ is
  bounded,
  \begin{equation}
    \label{eq:SAC-limit-bounded}
    \chi_{\rm SAC}(a)\le (1-\rho_0)^{-1}.     
  \end{equation}
\end{theorem}
Notice that this bound does not include limitations as in
Corollary \ref{th:subgraph-sequence-local}.  Generally, for an
increasing sequence, 
\begin{equation}
  \label{eq:limit-define}
  \rho_0\equiv \lim_{t\to\infty}\rho(H_p^{(t)})\le \rltwo(H_p).  
\end{equation}
 It may well
happen that $\rho_0<\rltwo(H_p)$, as in the case where
$\mathcal{D}$ is a tree, cf.\ Example \ref{ex:convergence-to-tree}.  
\begin{proof}[Proof of Theorem \ref{th:SAC-bound-infinite}]
  The sequence $\rho(H_p^{(t)})$ is non-decreasing. By assumption, it
  is also bounded, thus the limit exists.  For any $s$, the cycles of
  length $s$ are contained in a finite vicinity of the original arc
  $a$; the corresponding contribution does not exceed $\rho_0^s$ (see
  the proof of Theorem \ref{th:chi-SAC-bound}).
  Summation over $s$ gives the bound (\ref{eq:SAC-limit-bounded}).
\end{proof}

Combined with Theorem \ref{th:sac-divergent}, this gives the
following Corollary (its weaker version previously appeared as a
conjecture in Ref.~\cite{Hamilton-Pryadko-SIAM-2015}).

\begin{corollary}
  \label{th:nice-uniqueness-bound}
  Consider heterogeneous site percolation on an infinite, locally
  finite digraph $\mathcal{D}$, characterized by the weighted
  Hashimoto matrix $H_p$, with on-site probabilities
  $p_\nu\le p_\mathrm{max}<1$.  For an increasing sequence of finite
  subgraphs
  $\mathcal{D}^{(t)}\subset \mathcal{D}^{(t+1)}\subset\mathcal{D}$
  weakly converging to $\mathcal{D}$, let
  $\rho_0\equiv \lim_{t\to\infty}\rho(H_p^{(t)})$.  Then, if
  strong-cluster percolation happens while $\rho_0<1$, the percolating
  cluster cannot be unique.
\end{corollary}

If $\rho_0<1$ and yet an infinite strongly-connected cluster exists on
the induced subgraph $\mathcal{D}'$, the system is expected to be
below any transition associated with the number of percolating
clusters.  While such a transition is usually called ``uniqueness''
transition, we note that one-endedness of the infinite cluster is
necessary but not sufficient to have a unique percolating cluster.  An
example could be any graph with a finite number of ends. 

\begin{example}
  \label{ex:convergence-to-tree}
  A degree-$d$ infinite tree $\mathcal{T}_d$ can be obtained as a
  limit around its root of the following graph sequences: (\textbf{a})
  a sequence of its subgraphs, $t$-generation trees
  $\mathcal{T}_d^{(t)}$; (\textbf{b}) sequence of graphs obtained from
  $\mathcal{T}_d^{(t)}$ by pairing degree-one vertices arbitrarily to
  form degree-two vertices; (\textbf{c}) a sequence of $d$-regular
  graphs obtained from $\mathcal{T}_d^{(t)}$ by joining the leaves in
  pairs and replacing any resulting pair of edges connected by a
  degree-two vertex with a single edge.  In the case (\textbf{a}),
  $\rho(H^{(t)})=0$ for any $t$; in the case (\textbf{b}),
  $\lim_{t\to\infty}\rho(H^{(t)})=(d-1)^{1/2}=(\gr T_d)^{1/2}$, with
  the components of the PF vectors falling exponentially away from the
  center.  In the case (\textbf{c}), $\rho(H^{(t)})=d-1=\gr T_d$, and
  the PF vectors have all equal components.  The sequence of subgraphs
  (\textbf{a}) correctly reproduces (the absence of) the uniqueness
  transition, while the sequences of degree-regular graphs
  (\textbf{c}) give the percolation transition at $p_c=1/(d-1)$.
\end{example}
\begin{example}
  Consider a $(d,d)$-regular locally-planar hyperbolic graph
  $\mathcal{G}_d$, with $d$ identical $d$-sided plaquettes meeting at
  every vertex.  This graph can be obtained as a limit of (\textbf{a})
  an increasing sequence of radius-$t$ subgraphs, with
  $\lim_{t\to\infty}\rho(H^{(t)})=\rltwo(H)$, or (\textbf{b}) a
  sequence of $d$-regular graphs\cite{Delfosse-Zemor-2014} whose
  spectral radii coincide with $\gr H=\bargr H=d-1$.  
  The spectral radius of the adjacency matrix satisfies
  the following bounds based on the Cheeger
  constant\cite{Higuchi-Shirai-2003,Madras-Wu-2005},
  \begin{equation}
    \label{eq:Hd-bounds}
    2\sqrt{d-1}\le \rltwo(A)\le 2\sqrt{d}; 
  \end{equation}
these result in 
\begin{equation}
  \label{eq:4}
  \sqrt{d-1}\le \rltwo(H)\le 1+\sqrt{d}. 
\end{equation}
For
site percolation on $\mathcal{G}_d$, we recover the maximum degree
bound for percolation transition, $p_c=p_T\ge 1/(d-1)$, and get the
bound $p_u\ge 1/\sqrt{d-1}$ for the uniqueness transition.
\end{example}
\begin{example}
  For a $D$-dimensional hypercubic lattice, $\mathbb{Z}^D$, we have
  $\rltwo(H)=\gr H=\bargr H=2D-1=d-1$.  There is only one
  percolation transition, $p_c=p_u=p_T$.  Our  bounds for these
  transitions coincide and recover
  the maximum degree bound.
\end{example}

\begin{example}
  \label{ex:two-region-strong} Consider strongly-connected cluster
  percolation on the oriented graphs in Example~\ref{ex:two-region}.
  For any finite $L$, a strongly-connected cluster can be formed.
  Numerically, we get $p_c^\mathrm{(str)}\approx 0.53$, see
  Fig.~\ref{fig:numerics-str}.  This is a reasonable value since
  directed-cluster percolation in both regions is necessary to form a large
  strong cluster.  However, the weak limit of this graph sequence is
  an infinite oriented tree.  For such a tree, any strongly-connected
  cluster is limited to one site: strong-cluster percolation never
  happens.  Respectively, the limiting spectral radius in Corollary
  \ref{th:nice-uniqueness-bound} is $\rho_0=0$.
\end{example}
\begin{figure}[htbp] \centering 
  \includegraphics[width=0.9\columnwidth]{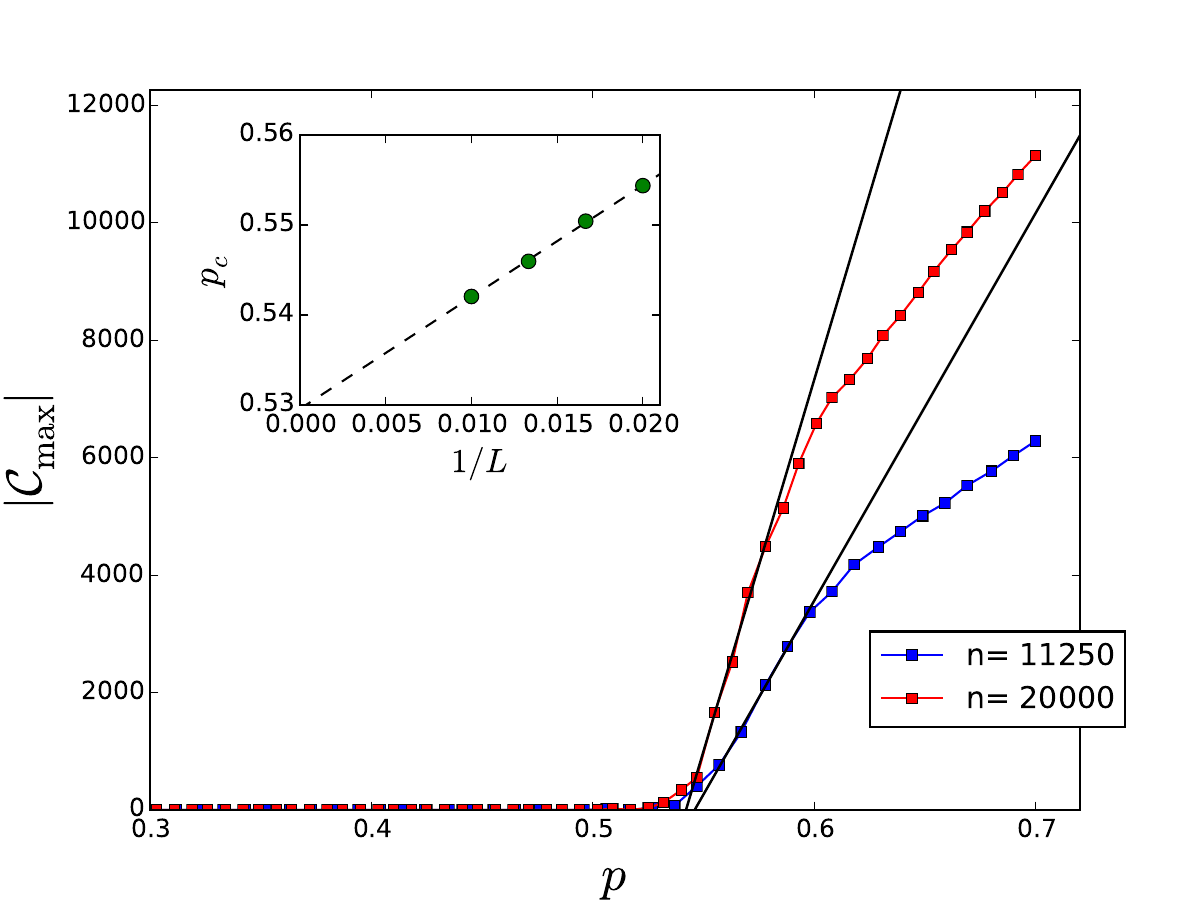}
  \caption{As in Fig.~\ref{fig:numerics} but for the largest
    strongly-connected cluster size as a function of on-site
    probability $p$, giving the strongly-connected cluster percolation
    transition at $p_c^{\rm (str)}=0.530\pm 0.01$.  Each point is an
    average over 120 realizations of the open vertex configurations
    with given site probability; the corresponding standard errors are
    smaller than the point size.  For this graph,
    all other strongly-connected clusters are all of size 1
    (single vertex).}
  \label{fig:numerics-str}
\end{figure}

\begin{example}
  \label{ex:core-with-leaves}
  For an integer $d_0>2$, consider site percolation on a graph
  $\mathcal{G}$ constructed from a well-connected core, a large random
  $d_0$-regular graph $\mathcal{G}_0$ with $n_0$ vertices, by
  connecting $r$ additional otherwise disjoint edges to each vertex in
  $\mathcal{V}(\mathcal{G}_0)$.  With $n_0$ large and $r$ bounded,
  formation of a giant component on ${\cal G}$ is governed by the
  corresponding transition on $\mathcal{G}_0$.  Spectral radii of the
  Hashimoto and the adjacency matrices of ${\cal G}$ are
  $\rho(H)=d_0-1$ and $\rho(A)=\Bigl[d_0+\left(4r+d_0^2\right)^{1/2}\Bigr]/2$.  We see that
  the bound in Corollary \ref{th:spectral-radius-generic-chi-bound}
  becomes increasingly loose as $r$ is increased in the region
  $r>d_0^2/4$.






\end{example}

\section{Discussion}
\label{sec:disc}
We constructed several bounds for heterogeneous percolation on general
graphs, directed or undirected.  We obtained explicit expressions for
finite and infinite (di)graphs, and, in two cases, analyzed the
continuity of the bounds in the infinite-graph limit.  Most bounds are
obtained from the non-back\-tracking path expansion, and formulated in
terms of appropriately weight\-ed Hashimoto matrices $H_p$.  While in
several cases stronger bounds (e.g., in terms of self-avoiding paths,
or using modified Hashimoto matrix as in
Ref.~\cite{Radicchi-Castellano-2016}) may be readily available, one
main advantage of the results presented here is that spectral radii
and norms can be calculated efficiently.

For a general infinite undirected graph, there are three transitions
usually associated with percolation: divergence of the cluster
susceptibility, formation of an infinite cluster, and the uniqueness
transition.  Bounds for all three transitions are formulated in terms
of the weighted Hashimoto matrix, see Theorems
\ref{th:out-cluster-infinite}, \ref{th:general-percolation-bound}, and
Corollary \ref{th:nice-uniqueness-bound}.  These bounds also apply for
directed- or strong-cluster percolation on digraphs.  In addition, the
condition $\rltwo(H_p)<1$ guarantees that connectivity on an infinite
digraph decays exponentially with the distance.

In practical network theory applications, more important is the
transition associated with the formation of a giant component.
Several simple criteria for the emergence of a giant component that
are commonly used in network theory rely on the degree distribution of
a graph. First is the lower bound for percolation on an arbitrary
graph in terms of the maximum degree \cite{Hammersley-1961}, see
Eq.~(\ref{eq:max-degree-bound}).  While it is universally applicable,
the issue with this inequality is that it tends to give very low
bounds on graphs with wide degree distribution.  Our Theorem
\ref{th:H-norm-one-bound} gives a generalization of this bound to
heterogeneous site percolation on arbitrary digraphs.  Second is the
Molloy-Reed criterion\cite{Gordon-1962,Molloy-Reed-1995,%
  Cohen-Erez-benAvraham-Havlin-2000,%
  Callaway-Newman-Strogatz-Watts-2000} which gives the percolation
threshold on random graphs in terms of the two first moments of degree
distribution.  It has been recently generalized to giant
in-/out-clusters on random digraphs\cite{Kryven-2016}.  While these
formulas are asymptotically exact in random graph
ensembles\cite{Chung-Horn-Lu-2009}, there is no guarantee: on actual
networks the Molloy-Reed criterion can substantially overestimate or
underestimate the threshold\cite{Radicchi-2015}.

The spectral radius has also been used to study percolation. On large
dense graphs, under mild conditions, the critical probability where a
giant cluster emerges is very close to the inverse spectral radius of
the adjacency matrix\cite{Bollobas-Borgs-Chayes-Riordan-2010}.  Our
Corollary \ref{th:undir-connectivity-bound} gives a related strict
bound for emergence of a giant strongly-connected component in
heterogeneous site percolation on an arbitrary digraph (a slightly
stronger bound specific for undirected graphs is given in Corollary
\ref{th:spectral-radius-generic-chi-bound}).  One substantial
advantage of these bounds is their universal applicability.  On the
other hand, bounds in terms of $\rho(A_p)$ can become loose in certain
carefully designed networks, see Example \ref{ex:core-with-leaves}.

In comparison, a bound in terms of the spectral radius of the
(weighted) Hashimoto matrix does not change upon the addition of 
leaves or finite trees.  It is also asymptotically exact for tree-like
graphs with few short cycles, large random graphs being the most
important example.  While such a bound does indeed limit the
\emph{percolation} transition on highly-uniform (e.g.,
quasi-transitive) (di)graphs, most generally the condition
$\rho(H_p)< 1$ is a bound on the strong-cluster \emph{uniqueness}
transition.  On an infinite digraph, such a bound is constructed as
the limit of the spectral radii for an increasing sequence of
subgraphs, which recovers the absence of the uniqueness transition on
an arbitrary infinite tree.  Further, for $\rltwo(H_p)<1$, on an
infinite digraph, connectivity decays exponentially with the distance.
On a finite digraph, with $\rho(H_p)<1$, strong connectivity also
decays exponentially with an additional local strong connectivity
condition needed to limit the prefactor in the bound.  Such an
exponential decay also implies sublinear scaling of the expected size
of the largest cluster, which in turn guarantees that a giant
strongly-connected component containing a non-zero fraction of all
vertices emerges with an asymptotically zero probability.

\section*{Acknowledgments}
We are grateful to M.\ E.\ J. Newman and N. Delfosse for enlightening
discussions, and to the anonymous Referee who insisted that we extend
the original analysis to infinite graphs.  This work was supported in
part by the U.S.\ Army Research Office under Grant No.\
W911NF-14-1-0272 and by the NSF under Grant No.\ PHY-1416578.  LPP
also acknowledges hospitality by the Institute for Quantum Information
and Matter, an NSF Physics Frontiers Center with support of the Gordon
and Betty Moore Foundation.

\bibliography{Hamilton_Pryadko}
\end{document}